\documentclass[sigconf, nonacm]{acmart}

\newcommand\vldbpagestyle{plain}

\usepackage{epsfig}
\usepackage{epstopdf}
\usepackage{cleveref}
\usepackage{mathtools}
\usepackage{algorithm}
\usepackage{subfig}
\usepackage{multirow}
\usepackage{dsfont}
\usepackage{varwidth}
\usepackage[noend]{algpseudocode}
\algrenewcommand\algorithmicrequire{\textbf{Input:}}
\algrenewcommand\algorithmicensure{\textbf{Output:}}

\DeclareMathSymbol{\mlq}{\mathord}{operators}{``}
\DeclareMathSymbol{\mrq}{\mathord}{operators}{`'}

\usepackage{graphicx}
\usepackage{booktabs}
\usepackage{amsmath}
\usepackage{array}
\usepackage{amsfonts}
\usepackage{epsfig}
\usepackage{graphicx}
\usepackage{url}
\usepackage[bottom]{footmisc}
\usepackage{balance}
\usepackage{multirow}
\usepackage{xspace}
\usepackage{tabularx}

\usepackage{pifont}
\usepackage{url}

\newtheorem{theor}{Theorem}

\newtheorem{exam}{Example}

\newcommand{\spara}[1]{\smallskip\noindent{\bf #1}}

\renewcommand\arraystretch{1.3}

\newcommand{\squishlist}{
 \begin{list}{$\bullet$}
  {  \setlength{\itemsep}{0pt}
     \setlength{\parsep}{3pt}
     \setlength{\topsep}{3pt}
     \setlength{\partopsep}{0pt}
     \setlength{\leftmargin}{2em}
     \setlength{\labelwidth}{1.5em}
     \setlength{\labelsep}{0.5em}
} }
\newcommand{\squishlisttight}{
 \begin{list}{$\bullet$}
  { \setlength{\itemsep}{0pt}
    \setlength{\parsep}{0pt}
    \setlength{\topsep}{0pt}
    \setlength{\partopsep}{0pt}
    \setlength{\leftmargin}{2em}
    \setlength{\labelwidth}{1.5em}
    \setlength{\labelsep}{0.5em}
} }

\newcommand{\squishdesc}{
 \begin{list}{}
  {  \setlength{\itemsep}{0pt}
     \setlength{\parsep}{3pt}
     \setlength{\topsep}{3pt}
     \setlength{\partopsep}{0pt}
     \setlength{\leftmargin}{1em}
     \setlength{\labelwidth}{1.5em}
     \setlength{\labelsep}{0.5em}
} }

\newcommand{\squishend}{
  \end{list}
}

\newcommand{\eat}[1]{}

\newcommand{\NP}{\ensuremath{\mathbf{NP}}\xspace}

\newcounter{ccc}

\DeclareMathOperator*{\argmax}{arg\,max}
\newcommand{\bigO}{\mathcal{O}}

\begin{document}
\title{Distributed Graph Embedding with Information-Oriented Random Walks}

\settopmatter{authorsperrow=4}
\author{Peng Fang}
\affiliation{%
  \institution{Huazhong University of Science and Technology}
  \country{China}
}
\email{fangpeng@hust.edu.cn}

\author{Arijit Khan}
\affiliation{%
  \institution{Aalborg University}
  \country{Denmark}
}
\email{arijitk@cs.aau.dk}

\author{Siqiang Luo}
\authornote{Corresponding author}
\affiliation{%
  \institution{Nanyang Technological University}
  \country{ Singapore}
}
\email{siqiang.luo@ntu.edu.sg}

\author{Fang Wang}
\authornotemark[1]
\affiliation{%
  \institution{Huazhong University of Science and Technology}
  \country{China}
}
\email{wangfang@hust.edu.cn}

\author{Dan Feng}
\affiliation{%
  \institution{Huazhong University of Science and Technology}
  \country{China}
}
\email{dfeng@hust.edu.cn}

\author{Zhenli Li} %
\affiliation{%
  \institution{Huazhong University of Science and Technology}
  \country{China}
}
\email{lizhenli@hust.edu.cn}

\author{Wei Yin}
\affiliation{%
  \institution{Huazhong University of Science and Technology}
  \country{China}
}
\email{weiyin_hust@hust.edu.cn}

\author{Yuchao Cao}
\affiliation{%
  \institution{Huazhong University of Science and Technology}
  \country{China}
}
\email{caoyuchao@hust.edu.cn}

\begin{abstract}
Graph embedding maps graph nodes to low-dimensional vectors, and is widely adopted in machine learning tasks.
The increasing availability of billion-edge graphs underscores the importance of learning efficient and effective embeddings on large graphs, such as link prediction on Twitter with over one billion edges. 
Most existing graph embedding methods fall short of reaching high data scalability. In this paper, we present a 
general-purpose, distributed,
information-centric random walk-based graph embedding framework, 
{\sf DistGER}, which 
can scale to embed billion-edge graphs. {\sf DistGER} incrementally computes
information-centric random walks. 
It further leverages a multi-proximity-aware,
streaming, parallel graph partitioning strategy, simultaneously achieving high local partition quality and excellent workload balancing across machines. {\sf DistGER} also improves the distributed {\sf Skip-Gram} learning model to generate node embeddings
by optimizing the access locality, CPU throughput, and synchronization efficiency.
Experiments on real-world graphs demonstrate that compared to state-of-the-art distributed graph embedding frameworks, including {\sf KnightKing}, {\sf DistDGL}, and {\sf Pytorch-BigGraph},
{\sf DistGER} 
exhibits $2.33\times$--$129\times$ acceleration, 45\% reduction in cross-machines communication,
and \textgreater 10\% effectiveness improvement in downstream tasks.
\end{abstract}


\maketitle

\pagestyle{\vldbpagestyle}
%
%

\section{Introduction}
\label{sec:intro}
Graph embedding is a widely adopted operation that embeds a node in a graph to a low-dimensional vector.
The embedding results are used in downstream machine learning tasks 
such as link prediction \cite{Link_prediction_2017},
node classification \cite{classification_2011}, clustering \cite{cluster_2017}, and recommendation \cite{recomendation_2017}.
In these applications, graphs can be huge with millions of nodes and billions of edges. For instance, the Twitter graph includes over 41 million user nodes and over one billion edges, and it has extensive requirements for link prediction and classification tasks \cite{GuptaGLSWZ13}. The graph of users and products at Alibaba also consists of more than two billion user-product edges, which forms a giant bipartite graph for its recommendation tasks \cite{WangHZZZL18}. 
A plethora of random walk-based graph embedding solutions \cite{HuGE_2021,node2vec_2016,DeepWalk_2014,Line_2015,Verse_2018}
are proposed. 
A random walk is a graph traversal that starts from a source node,
jumps to a neighboring node at each step, and stops after a few steps.
Random-walk-based embeddings are inspired by the well-known natural
language processing model, {\sf word2vec} \cite{word2vec_2013}. By conducting sufficient
random walks on graphs, substantial graph structural information is collected and fed into the
{\sf word2vec} ({\sf Skip-Gram}) to generate node embeddings. Compared with other graph embedding
solutions such as graph neural networks \cite{TuCWY018,WangC016,GraphSAGE_2017,Graph_attention_2018,Graphgan_2018}
and matrix factorization techniques \cite{ProNE_2019,NetSMF_2019,NRP_2020,QiuDMLWT18,WangCWP0Y17},
random walk-based methods are more flexible, parallel-friendly, and scale to larger graphs \cite{FlashMob_2021}.

While graph embedding is crucial, the increasing availability of billion-edge
graphs underscores the importance of scaling graph embedding.
The inherent challenge is that the number of random walks required
increases with the graph size. For example, one representative work, {\sf node2vec}
\cite{node2vec_2016} needs to sample many node pairs to ensure the embedding quality, 
it takes months to learn node embeddings for a graph with 100 million nodes and 500 million edges by 20 threads on a modern server \cite{ProNE_2019}.
Some very recent work, e.g., {\sf HuGE} \cite{HuGE_2021} attempts to improve the quality of random walks 
according to the importance of nodes. Though this method can remove redundant
random walks to a great extent, the inherent complexity remains similar.
It still requires more than one week to learn embeddings for a billion-edge Twitter graph on a modern server,
hindering its adoption to real-world applications. Another line of work turns
to use GPUs for efficient graph embedding. For example, some recent graph embedding frameworks
(e.g., \cite{GraphVite_2019, Tencent_GE_2020}) simultaneously perform graph random walks on CPUs and
embedding training on GPUs. However, as the computing power between GPUs and CPUs differ widely,
it is typically hard for the random walk procedure performed on CPUs to catch up with the embedding
computation performed on GPUs, causing bottlenecks \cite{Serafini21,ZhengCCSWLCYZ22}. Furthermore, this process is heavily related to
GPUs' computing and memory capacity, which can be drastically different across different servers.

Recently, distributed graph embedding, or computing graph embeddings with multiple machines, has attracted significant research interest
to address the scalability issue. Examples include {\sf KnightKing} \cite{KnighKing_2019}, {\sf Pytorch-BigGraph} \cite{PBG_2019},
and {\sf DistDGL} \cite{DistDGL_2020}. {\sf KnightKing} \cite{KnighKing_2019} optimizes the walk-forwarding process for {\sf node2vec} and brings up several orders of magnitude improvement compared to a single-server solution. However, it may suffer from
redundant or insufficient random walks that are attributed to a routine random walk setting, resulting in low-quality training information for the
downstream task \cite{HuGE_2021}. Moreover, the workload-balancing graph partitioning scheme that it leverages fails
to consider the randomness inherent in random walks, introducing higher communication costs across machines
and degrading its performance. Facebook proposes {\sf Pytorch-BigGraph} \cite{PBG_2019} that leverages
graph partitioning technique and parameter server to learn large graph embedding on multiple CPUs based on {\sf PyTorch}.
However, the parameter server used in this framework needs to synchronize embeddings with clients,
which puts more load on the communication network and limits its scalability.
Amazon has recently released {\sf DistDGL} \cite{DistDGL_2020}, a distributed graph embedding framework for graph neural network model.
However, its efficiency is bogged down by the graph sampling operation, e.g., more than 80\% of the overhead is for sampling in the
{\sf GraphSAGE} model \cite{GraphSAGE_2017}, and the mini-batch sampling used may trigger delays in gradient updates causing
inefficient synchronization. In conclusion, although the distributed computation frameworks have shown better performance
than the single-server and CPU-GPU-based solutions, significant rooms exist for further improvement.

{\bf Our {\sf DistGER} system.}
We present a newly designed distributed graph embedding system, {\sf DistGER},
which incorporates more effective information-centric random walks such as {\sf HuGE}~\cite{HuGE_2021}
and achieves super-fast graph embedding compared to state-of-the-arts.
As a preview, compared to {\sf KnightKing}, {\sf Pytorch-BigGraph}, and {\sf DistDGL},
our {\sf DistGER} achieves $9.3\times$, $26.2\times$, and $51.9\times$ faster embedding on average, and easily scales to
billion-edge graphs (\S \ref{sec:experiments}).
Due to information-centric random walks, {\sf DistGER} embedding also shows higher effectiveness when applied to downstream tasks.

Three novel contributions of {\sf DistGER} are as follows.
{\bf First} and foremost, since the information-centric random walk requires measuring the effectiveness of the
generated walk on-the-fly during the walking procedure, it inevitably introduces higher
computation and communication costs in a distributed setting.
{\sf DistGER} resolves this by showing that the effectiveness of a
walking path can be measured through incremental information,
avoiding the need for full-path information.
{\sf DistGER} invents incremental information-centric computing ({\sf InCoM}),
ensuring $\bigO(1)$ time for on-the-fly measurement and maintains constant-size messages
across computing machines. {\bf Second}, considering the randomness inherent in random walks
and the workload balancing requirement, {\sf DistGER} proposes multi-proximity-aware,
streaming, parallel graph partitioning ({\sf MPGP}) that is adaptive to random walk characteristics, increasing
the local partition utilization. Meanwhile, it uses a dynamic workload constraint for the partitioning
strategy to ensure load-balancing. {\bf Finally}, different from the existing random walk-based embedding techniques,
{\sf DistGER} designs a distributed {\sf Skip-Gram} learning model ({\sf DSGL}) to generate node embeddings
and implements an end-to-end distributed graph embedding system. Precisely, {\sf DSGL} leverages global-matrices and two local-buffers
for node vectors to improve the access locality, thus reducing cache lines ping-ponging across multiple cores
during model updates; then develops multi-windows shared-negative samples computation to fully exploit the CPU throughput.
Moreover, a hotness block-based synchronization mechanism is proposed to synchronize node vectors efficiently in a distributed setting.

{\bf Our contributions and roadmap.}
We propose an efficient, scalable, end-to-end distributed graph embedding system, {\sf DistGER},
which, to our best knowledge, is the first general-purpose, information-centric random walk-based
distributed graph embedding framework.

\noindent
$\bullet$ We introduce incremental information-centric computing ({\sf InCoM}) to address computation and communication
overheads due to on-the-fly effectiveness measurements during information-oriented random walks in a distributed setting (\S \ref{sec:incom}).

\noindent
$\bullet$ We propose multi-proximity-aware, streaming, parallel graph partitioning ({\sf MPGP}) that 
achieves both higher local partition utilization and load-balancing (\S \ref{sec:partition}).

\noindent
$\bullet$ We develop a distributed {\sf Skip-Gram} learning model ({\sf DSGL}) to generate node embeddings by improving
the access locality, CPU throughput, and synchronization efficiency (\S \ref{sec:learning}).

\noindent
$\bullet$ We conduct extensive experiments on five large, real-world graphs to 
demonstrate that {\sf DistGER} achieves much better efficiency, scalability, and effectiveness over existing popular distributed frameworks,
e.g., {\sf KnightKing} \cite{KnighKing_2019}, {\sf DistDGL} \cite{DistDGL_2020}, and {\sf Pytorch-BigGraph} \cite{PBG_2019}. 
In addition, {\sf DistGER} generalizes well to other random walk-based graph embedding methods (\S \ref{sec:experiments}). 

We discuss preliminaries and a baseline approach in \S \ref{sec:preliminaries}, related work in \S \ref{sec:related}, and conclude in \S \ref{sec:conclusions}. 

\section{Preliminaries and Baseline}
\label{sec:preliminaries}
We design an end-to-end distributed system for effective and scalable embedding of large graphs via random walks.
To this end, we first discuss relevant works on random walk-based sequential graph embedding (\S \ref{sec:randGembedding}) and
distributed systems for random walks on graphs (\S \ref{sec:dRand}). Then, we propose a baseline distributed system for
random walk-based graph embedding by combining the above two methods (\S \ref{sec:HUGED}), discuss its limitations
and scopes of improvements, which leads to introducing our ultimate system, {\sf DistGER}
in \S \ref{sec:DistGER} and \S \ref{sec:learning}.
Table~\ref{tab:Notations_paper} explains the most important notations.
{\sf DistGER} handles undirected and unweighted graphs by default,
but can support directed and weighted  graphs (higher edge weights imply stronger connectivity). {\sf DitsGER} uses the {\em Compressed Sparse Row} (CSR) \cite{csr_format}
format to store graph data, where directed edges are stored with their source nodes and undirected edges are stored twice for both directions.
For each weighted edge, CSR stores a tuple containing its destination node and edge weight.
\begin{table}[tb!]
\footnotesize
\centering
\begin{center}
  \caption{\small Frequently used notations}
  \label{tab:Notations_paper}
  \begin{tabular}{l||l} 
  \textbf{Notation}   & \textbf{Meaning} \\ \hline \hline
    $G = (V, E)$ & $G$: undirected, unweighted graph; $V$: set of nodes; $E$: set of edges\\
    $\varphi(u)$ & embedding or vector representation of node $u$, having dimension $d$ \\
    $w$          & window size of context in the {\sf Skip-Gram}\\
    $N(u)$       & neighbors of node $u$ \\
    $L$          & random walk length starting from a node\\
    $r$          & number of random walks per node\\
    $H(X)$       & entropy of random variable $X$ with possible values {$x_1, x_2, \ldots, x_n$} \\
\end{tabular}
\end{center}
\end{table}
\subsection{Random-walks Based Graph Embedding}
\label{sec:randGembedding}
These graph embedding algorithms are inspired by the well-known natural language processing model, {\sf word2vec} \cite{word2vec_2013}:
They transform a graph into a set of random walks through sampling methods, treat each random walk as a sentence,
and then adopt {\sf word2vec} ({\sf Skip-Gram}) to generate node embeddings from the sampled walks.

\spara{Node2vec.}
A most representative algorithm in the aforementioned category is {\sf node2vec} \cite{node2vec_2016}, as given below.

\underline{Random walk method.}
Given a graph $G=(V, E)$, two nodes $u, v \in  V$, and we suppose a walker is currently at node $u$.
{\sf Node2vec} defines the transition probability from $u$ to $v$ as $P(u,v)= \frac{\pi_{uv}}{Z}$,
%
%
where $\pi_{uv}$ is the unnormalized transition probability from $u$ to $v$, and $Z$ is the normalization constant defined as $\sum_{v\in N(u)} \pi_{uv}$.
{\sf Node2vec} defines a second-order random walk. 
Assume that a walker just traversed node $t$ and now resides at node $u$ ($u$ is a neighbor of $t$).
Next, it will select a node $v$ from $u$'s neighbors.
The un-normalized transition probability $\pi_{uv}$
is defined by $d_{tv}$, which is the shortest path distance between nodes $t$ and $v$:
If $d_{tv}$ is 0, 1, 2, respectively, then the corresponding $\pi_{uv}$ is $1/p$, $1$, $1/q$.
%
Hyperparameters $p$ and $q$ are called return and in-out parameters, respectively.
$d_{tv} = 0$ means that $t$ and $v$ are the same node, i.e., the
walker goes back to $t$, which is a BFS-like exploration, thus setting a small $p$ obtains
a ``local view" in the graph with respect to the start node.
$d_{tv} = 1$ means that $v$ is a neighbor of $t$, 
and $d_{tv} = 2$ denotes a DFS-like exploration 
to get a ``global view" in the graph, which can be attained by a small $q$.

\underline{Features learning for graph embedding.}
Features learning maps $\varphi: V \to R^d $ from nodes to feature representations (node embeddings).
Since {\sf node2vec} captures node representations based on the {\sf Skip-Gram} model \cite{word2vec_2013}
that maximizes the co-occurrence probability
between words within a window $w$ in a sentence, 
the objective is:
\begin{small}
\begin{equation}
\operatorname*{argmax}_\varphi \frac{1}{|V|}\sum\limits_{j=1}^{|V|}\sum\limits_{-w\leq i\leq w}\log{p(u_{j+i}| u_j)}
\label{skim_gram_eq}
\end{equation}
\end{small}
The generated walks are used as a corpus with vocabulary $V$,
where $u_{j+i}$ denotes a context node in a window $w$, and $p(u_{j+i}| u_j)$ indicates the probability to predict the context node.
The basic {\sf Skip-Gram} formulates $p(u_j| u_{j+i})$ as the softmax function.
%
%
%
%
%
Existing methods generally
speed-up training 
with negative sampling \cite{negative_sampling_2013}.
\begin{small}
\begin{equation}
\label{negative_sampling_Eq}
\begin{aligned}
\log p(u_j| u_{j+i}) &\approx \log\sigma(\varphi_{in}(u_{j+i})\cdot\varphi_{out}(u_{j})) \\
&+\sum\limits_{k=1}^{K}\mathbb{E}_{u_k \sim Pn(u)}[\log\sigma(-\varphi_{in}(u_{j+i})\cdot\varphi_{out}(u_{k}))]
\end{aligned}
\end{equation}
\end{small}
Here, $\sigma(x)=\frac{1}{1+exp(-x)}$ is the sigmoid function, and the expectations
are computed by drawing random nodes from a sampling distribution $Pn(u)$, $\forall u \in V$.
Typically, the number of negative samples $K$ is much smaller than $|V|$ (e.g., $K\in[5,20]$).

\underline{Complexity analysis.}
Assume that the number of walks per node is $r$, walk length $L$, embedding dimensions $d$, window size $w$,
and the number of negative samples $K$. The time complexity of {\sf node2vec} random-walk procedure
is $\bigO(r\cdot L\cdot |V|)$.
For feature learning, the corpus size $C = r\cdot L$. Let us denote the complexity of the unit operation of predicting and updating one node's
embedding as $o$. The {\sf Skip-Gram} with the negative sampling only needs $K+1$ words to obtain a probability distribution
(Eq.~\ref{negative_sampling_Eq}), thus the time complexity of
{\sf node2vec} feature learning is $\bigO(C \cdot w \cdot (K+1) \cdot o)$. 
Since each node in the {\sf Skip-Gram} model needs to maintain two embeddings $\varphi_{in}$ and $\varphi_{out}$
for the parameter updates, the space complexity of {\sf node2vec}, which refers to the parameter sizes, is $O(|V|d)$.

\underline{Drawbacks.}
Despite the flexibility in exploring node representations (local-view vs. global-view),
{\sf node2vec} incurs high time overhead. 
It leverages a routine random walk configuration (usually, $L$=80 and $r$=10)
to generate walks, similar to most existing random walk-based graph embedding methods, which limits the efficiency and scalability on large-scale graphs. Indeed, this {\em one-size-fits-all}
strategy cannot meet the specific requirements of different real-world graphs. 
For instance, the high-degree nodes are usually located in dense areas of a graph, they might require longer and more random walks to capture more comprehensive features; while for the low-degree nodes, if treated equally, it may introduce redundancy into generated walks, thus limiting the scalability.
%

\spara{HuGE.}
The recent work, {\sf HuGE} \cite{HuGE_2021} attempts to resolve the routine random walk issue of {\sf node2vec} and proposes a novel information-oriented random walk mechanism to achieve a concise and comprehensive representation in the sampling procedure.

\underline{Random walk method.}
First, {\sf HuGE} leverages a hybrid random walk strategy, which considers both node degree and the number of common neighbors in each walking step. 
Common neighbors represent potential information between nodes, e.g., node similarity \cite{ Line_2015}. 
For random walks, high-degree nodes are revisited more, and walks starting from them can obtain richer information by traveling around their local neighbors \cite{ common_neighbor_aware_icde_2019}.
The un-normalized transition probability from node $u$ to the next-hop node $v$ is:
\begin{small}
\begin{equation}
\alpha(u,v) = \frac{1}{deg(u)-Cm(u, v)} \times \max \left\{\frac{deg(u)}{deg(v)}, \frac{deg(v)}{deg(u)} \right\}
\label{accept_CNHRW}
\end{equation}
\end{small}
where $deg(u)$ is the degree of $u$, and $Cm(u, v)$ denotes the number of their common neighbors. Thus, $\frac{1}{deg(u)-Cm(u, v)}$ indicates the similarity between the current node $u$ and the next-hop node $v$, the ratio grows with higher $Cm(u, v)$, since $deg(u)$ is fixed. The $\max$ function assigns a weight to the transition probability from $u$ to $v$, indicating the influence of a high degree node on its neighbors.
%

At the current node $u$, {\sf HuGE} randomly chooses $v$ from $N(u)$ as a candidate node, the acceptance probability for $v$ as the next-hop node is $P(u,v)$, and if $v$ is rejected, which happens with probability $1-P(u,v)$, the walker backtracks to $u$ and repeats a random selection again from $N(u)$, known as the {\em walking-backtracking} strategy \cite{ common_neighbor_aware_icde_2019}. $P(u,v)$ is defined as $Z\left(\alpha(u,v)\right)$,
%
%
where {\sf HuGE} normalizes $\alpha(u,v)$ via $Z(x) = \frac{e^x-e^{-x}}{e^x+e^{-x}}$, which is widely-applied in machine learning. 
For edge weight $w(u,v)$, we define $P(u,v)=Z\left(\alpha(u,v) \cdot w(u, v)\right)$.

Second, in contrast to a one-size-fits-all strategy, {\sf HuGE} proposes a heuristic walk length strategy to measure the effectiveness of information during walk based on entropy ($H$).
Mathematically, let us denote the random walk starting at the source node $u$ as $W_u^L = \{ v_u^1, v_u^2, v_u^3, \ldots,v_u^L \}$, where $v_u^k$ denotes the $k$-th node on the walk. The probability of
the occurrence of a specific node $v$ on the walk is $\frac{n(v)}{L}$, where $n(v)$ is the number of occurrences of $v$ on the generated walk. The information entropy of the generated walk is:
\begin{small}
\begin{equation}
H\left(W_u^L\right) = -\sum_{v \in W_u^L}\frac{n(v)}{L}\log\frac{n(v)}{L}
\label{path_entropy}
\end{equation}
\end{small}
With increasing $L$, as the occurrence probability for a specific node in a generated walk gradually stabilizes, $H\left(W_u^L\right)$ initially grows with $L$ until it converges. {\sf HuGE} characterizes the correlation between $H\left(W_u^L\right)$ and $L$ by linear regression and calculates the coefficient of determination ($R^2$) to determine the termination of a random walk.
\begin{footnotesize}
\begin{equation}
R\left(H(W_u^L),L\right) = \frac{\sum_{i=1}^{L^*}\left(H\left(W_u^{L(i)}\right)-\overline{H(W_u^L)}\right)\left(L(i)-\overline{L}\right)}{\sqrt{\sum_{i=1}^{L^*}\left(H\left(W_u^{L(i)}\right)-\overline{H(W_u^L)}\right)^2}\sqrt{\sum_{i=1}^{L^*}\left(L(i)-\overline{L}\right)^2}}
\label{path_corr}
\end{equation}
\end{footnotesize}
$\overline{H\left(W_u^L\right)}$ and $\overline{L}$ are the mean of the respective series for $1 \leq i\leq L^*$, and $L^*$ is the optimal walk length for the current walk. As $L$ grows and $H(W_u^L)$ stabilizes, $R^2{(H,L)}$ also decreases and converges to 0, since their linear correlation diminishes. {\sf HuGE} sets  $R^2(H,L) < \mu$ as the walk termination condition. 
Setting a smaller $\mu$ generates longer walks, 
introducing redundant information; while too large $\mu$ may not ensure good coverage of graph properties during sampling, since too short walks are generated. Based on our experimental results, good quality walk lengths are attained with $\mu=0.995$.

Third, {\sf HuGE} also proposes a heuristic number of walks strategy. The corpus is generated by multiple ($r$) random walks from each node. Following \cite{DeepWalk_2014}, if the degree distribution of a connected graph follows power-law, the frequency in which nodes appear in short random walks will also follow a power-law distribution. Inspired by this observation, {\sf HuGE} empirically analyzes the similarity between the two distributions via relative entropy. 
Formally, the node degree distribution is expressed as $p(v) = \frac{deg(v)}{\sum_{v\in V}deg(v)}$. We denote the number of occurrences of $v$ in the generated corpus as $ocn(v)$. The probability distribution for such appearances in the corpus is $q(v) = ocn(v)/\sum_{v\in V}ocn(v)$. The relative entropy from $p$ to $q$ is:
\begin{small}
\begin{equation}
\begin{aligned}
D(p\rVert q)
&=\sum_{i=1}^{r^*}\frac{deg(v)}{\sum deg(v)}\log\frac{{deg(v)}{\sum ocn(v)}}{{ocn(v)}{\sum deg(v)}}
\end{aligned}
\end{equation}
\end{small}
Here, $r^*$ is the optimal number of walks from a source node. With increasing
$r$, the difference $D_r(p \rVert q)$ gradually converges, which means that the probability distribution of nodes' occurrences in the generated corpus has stabilized.
\begin{small}
\begin{equation}
\Delta D_r(p \rVert q)= |D_r(p\rVert q)-D_{r-1}(p\rVert q)|
\end{equation}
\end{small}

{\sf HuGE} leverages $\Delta D_r(p \rVert q) \leq \delta$ as the termination condition. Based on our experimental results, $\delta=0.001$ usually produces
a good number of random walks per source node.
%
%

\underline{Features learning for graph embedding.} The features learning uses the {\sf Skip-Gram} model, and similar to {\sf node2vec}, 
follows Eq.~\ref{skim_gram_eq} and~\ref{negative_sampling_Eq}.

\underline{Complexity analysis.}
The time complexity of {\sf HuGE} random-walk procedure is $\bigO(r'\cdot L'\cdot |V|)$, where
the optimal number of walks per node is $r'$ (decided by $\Delta D_r(q \rVert p) \leq \delta$) and
the average walk length is $L'$ (decided by $R^2(H,L) < \mu$ for each walk).
However, the complexity of measuring $H(W)$ and $R\left(H(W),L\right)$ at each step of a walk
is $\bigO(L)$, where $L$ is the current walk length.
Thus, the overall computational workload of {\sf HuGE}
becomes quadratic in the walk length, though the average walk length
can be smaller than that in {\sf Node2Vec}.
The time complexity of the feature learning phase remains $\bigO(C' \cdot w \cdot (K+1) \cdot o)$,
but the average corpus size $C'=r'\cdot L'$ is smaller, since generally $r'<r$ and $L'<L$.


\underline{Drawbacks.} The workload of {\sf HuGE} is quadratic in the walk length. Moreover, {\sf HuGE} \cite{HuGE_2021} embedding method is sequential, and there is yet no end-to-end distributed system to support graph embedding via information-oriented random walks. Being sequential, {\sf HuGE} requires more than one week to learn embeddings for a billion-edge Twitter graph on a modern server.
\subsection{Distributed Random Walks on Graphs}
\label{sec:dRand}
{\sf KnightKing} \cite{KnighKing_2019} is a recent general-purpose, distributed graph random walk engine. Its key components are introduced below.
\begin{algorithm} [tb!]
\footnotesize
\caption{\small {\sf HuGE-D} walking procedure}
\label{HuGE-D_walk}
\begin{algorithmic}[1]
\Require current node $u$, candidate node $v$, Walker $W$, {\sf HuGE} parameter $\mu$
\Ensure walker state updates
{\flushleft{{\bf sendStateQuery($u$, $v$, $W$)}}} 
\State{$P(u,v) = Z\left(\frac{1}{deg(u)-Cm(u, v)}\cdot \max \left\{\frac{deg(u)}{deg(v)}, \frac{deg(v)}{deg(u)} \right\} \right)$} // Eq.~\ref{accept_CNHRW}

{\flushleft{{\bf getStateQueryResult($W, P(u,v)$)}}} 
\State{generate a random number $\eta \in\left[0,1\right]$}
\If{$P(u,v)> \eta $}
\State{$W.path$.append($v$), $W.cur$ = $v$, $W.steps$ ++}
\State{$L$ = $W.steps$}
\State compute $H(W)$ and $R\left(H(W),L\right)$ // Eq.~\ref{path_entropy}, \ref{path_corr}
\If{$R^2(H(W),L) < \mu$}
\State terminate the walk
\Else
\State{generate another candidate node $t$ of $v$}
\State{sendStateQuery($v$, $t$, $W$)}
\EndIf
\Else
\State{backtrack to $u$ and generate another candidate node $v'$ of $u$}
\State{sendStateQuery($u$, $v'$, $W$)}
\EndIf
\end{algorithmic}
\end{algorithm}

\underline{Walker-centric programming model.}
For higher-order walks (e.g., second-order random walks in {\sf node2vec}),
{\sf KnightKing} assumes a walker-centric view. 
{\sf KnightKing} implements each step of 
walks by two rounds of message passing,
one round for walkers to submit the walker-to-node query messages
to check the distance between the previous node $W.prev$ and the candidate node $v$ 
and to generate the un-normalized transition probability $\pi_{uv}$;
another round of message passing
returns results to walkers about their states, and then
the walkers decide the next steps based on the sampling outcome.

{\sf KnightKing} coordinates many walkers simultaneously
based on the Bulk Synchronous Parallel ({\sf BSP}) model \cite{BSP_1990}. Walkers are assigned to some computing machine/thread. 
{\sf KnightKing} leverages rejection sampling to 
eliminate the need of scanning all out-edges at the walker's current node.  
Suppose a walker is currently at node $u$, the sampling method generates a candidate node $v$ (a neighbor of $u$) with a transition probability $p(u,v)$. The key idea of rejection sampling is to find an envelop $Q(u)=max(\frac{1}{p}, 1, \frac{1}{q})$: 
Within the rectangular area covered by the lines $y =Q(u)$ and $x =|N(u)|$,
it randomly samples a location $(x, y)$ from this area,  if $p(u,v) \geq y$, $v$ is accepted as a successfully sampled node,
otherwise, $v$ is rejected, and the method conducts more sampling trials until success. 
%
\begin{figure*}
   \centering
   \includegraphics[width= 6 in]{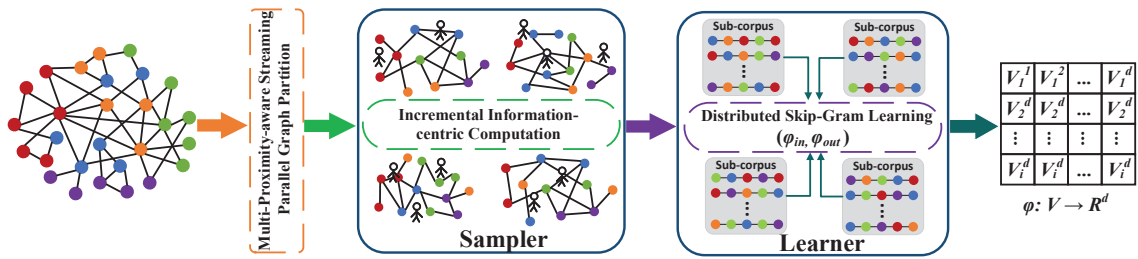}
   \caption{\small The workflow of our proposed system: {\sf DistGER}}
   \label{HuGER_framework}
 \end{figure*}

\underline{Workload-balancing graph partition.}
{\sf KnightKing} adopts a node-partitioning scheme -- each node (together with its edges)
is assigned to one computing machine in a distributed setting.
It roughly estimates the workload as the sum of the number 
of edges in each
computing machine, and ensures a balance of workloads across computing machines by
appropriately distributing the nodes. 
%
%

\underline{Complexity analysis.} For walk forwarding, 
each trial of rejection sampling 
needs $\bigO(1)$ time. 
With a reasonable $Q(u)$, there is a good chance for the sampling to succeed within a few trials,
thus the walk forwarding computation for one step can
be achieved in near $\bigO(1)$ complexity.
For communication, each message consists of
[$walk\_id$, $steps$, $node\_id$, $previous\_node\_id$] for {\sf node2vec}.
When a walk crosses a computing machine, it sends $M(1)$, i.e.,
one constant-length message to another machine. In a distributed environment,
assuming $P$ processors and the network bandwidth $B$, as analyzed in {\sf node2vec},
the total workload for {\sf KnightKing} is $\bigO(r\cdot L\cdot |V|)$, with near $\bigO(1)$ computation
complexity for each step. Thus, the average time spent in each processor is
$\bigO(r\cdot L\cdot |V|/P)$. The communication cost is $\bigO(N\cdot M(1)/B)$,
where $N$ is the count of cross-machine messages.
Thus, the time spent for {\sf KnightKing} 
is: $\bigO(r\cdot L\cdot |V|/P + N\cdot M(1)/B)$.

\underline{Drawbacks.} 
{\sf KnightKing}
provides distributed system support for traditional random walks, e.g., the one in {\sf node2vec}. For information-oriented random walks in {\sf HuGE}, the walkers need to additionally maintain the generated walking path at each step, 
and thus the message requires to carry the path information. 
The message length increases with the length of the walks, thus the efficiency of {\sf KnightKing} reduces due to extra overheads of computation and communication (elaborated in \S \ref{sec:HUGED}). 
The purely load-balancing partition scheme in {\sf KnightKing} also introduces high communication cost. 
%
\subsection{Baseline: HuGE-D}
\label{sec:HUGED}
Our distributed baseline approach, {\sf HuGE-D} 
replaces the traditional random walking method in {\sf KnightKing} with the information-oriented scheme of {\sf HuGE},
and leverages a full-path computation mechanism. 
%
%

\underline{Full-path computation mechanism.}
As previously stated, to meet the information measurement requirements,
{\sf HuGE-D} stores the full-path in the message, in addition to the necessary fields
such as $walk\_id$, $steps$, and $node\_id$. Algorithm~\ref{HuGE-D_walk} shows the
walking procedure, where we retain the {\em walking-backtracking} strategy of
{\sf HuGE}, which is similar to rejection sampling in {\sf KnightKing}.

\underline{Complexity analysis.}
{\sf HuGE-D} measures the information effectiveness of generated walks at each step.
The complexity of measuring $H(W)$ and $R\left(H(W), L\right)$ at each step
is $\bigO(L)$, where $L$ is the current walk length. For the communication cost due to messages, {\sf HuGE-D}
additionally requires carrying the generated walking path information in contrast to {\sf KnightKing},
so the message cost is also linear in the walk length $L$, which we denote as $M(L)$.
Similar to {\sf KnightKing}, since the total workload is related to the average walk length $L'$
and the optimal number of walks per node $r'$, the time spent for {\sf HuGE-D} in distributed setting is:
$\bigO(r'\cdot (L')^2\cdot |V|/P + N\cdot M(L')/B)$,
where $P$, $B$, and $N$ are the number of processors, network bandwidth,
and the count of cross-machine messages, respectively.
%

\underline{Drawbacks.}
{\bf (1) Computation:} 
{\sf HuGE-D} measures the effectiveness of generated walks at each step, which
has $\bigO(L)$ complexity, where $L$ is the current walk length. Thus, unlike {\sf KnightKing}, the computational workload of {\sf HuGE-D}
is quadratic in walk length.
%
{\bf (2) Communication:} {\sf KnightKing} only sends constant-length messages, e.g.,
for {\sf node2vec} each message has [$walk\_id$, $steps$, $node\_id$, $prev\_node\_id$],
but the messages in {\sf HuGE-D} carry the full-path
information (i.e., [$walk\_id$, $steps$, $node\_ id$, $path\_info$])
for information measurements, thus the message cost is linear in the walk length.
{\bf (3) Partitioning:} The workload-balancing partition in {\sf KnightKing}
fails to consider the large amount of cross-machine communications introduced by the randomness inherent in random walks.
%
%

\section{The Proposed System: D\lowercase{ist}GER}
\label{sec:DistGER}
To address the aforementioned computing and communication challenges of the baseline {\sf HuGE-D} (\S \ref{sec:HUGED}),
we ultimately design an efficient and scalable distributed information-centric random walks engine,
{\sf DistGER}, aiming to provide an end-to-end support for distributed random walks-based graph embedding
(Figure \ref{HuGER_framework}).
We discuss our distributed information-centric random walk component ({\sf sampler}) in \S \ref{sec:incom}
and multi-proximity-aware, streaming, parallel graph partitioning scheme ({\sf MPGP}) in \S \ref{sec:partition},
while our novel distributed graph embedding ({\sf learner})
is given in \S \ref{sec:learning}.
 Besides providing systemic support for the information-oriented method {\sf HuGE},
{\sf DistGER} can also extend its information-centric measurements to
traditional random walk-based approaches via the general API to get rid of their routine configurations (\S \ref{sec:generality}).
\subsection {Incremental Information-centric Computing}
\label{sec:incom}
{\sf DistGER} introduces incremental information-centric computing ({\sf InCoM})
to reduce redundant computations and the costs of messages. Recall that the baseline {\sf HuGE-D} computes $H(W)$ and $R(H(W), L)$ via the full-path computation mechanism to measure the information effectiveness of a walk $W$ at each step, requiring $\bigO(L)$ time at every step of the walk.
We show that instead it is possible to update $H(W)$ and $R(H(W),L)$ incrementally with $\bigO(1)$
time cost at every step of the walk. We store local information about ongoing walks
at the respective machines, which reduces cross-machine message sizes.

\underline{Incremental computing of walk information.}
Each computing machine stores a partition of nodes from the input graph.
For every ongoing walk $W$, each machine additionally maintains
in its {\em local frequency list} the set of nodes
$v$ present in the walk which are also local to that machine, together with their number of occurrences $n(v)$ on the walk.
When $W$ terminates, the local frequency list for $W$ is no longer required and is deleted.
Theorem~\ref{th:incom_walk} summarizes our incremental computing of walk information.
%
\begin{theor}
\label{th:incom_walk}
Consider an ongoing walk $W^L$ with the current length $L\geq 0$, the next accepted node to be added in $W^L$ is $v$, and
$n(v)\geq 0$ is the number of occurrences of $v$ in the walk. In addition to $v$, both $L$ and $n(v)$ would
increase by 1. For clarity, we denote $n(v)$ in $W^L$ and $W^{L+1}$ as
$n^L(v)$ and $n^{L+1}(v)$, respectively. The information entropy $H\left(W^{L+1}\right)$ is related to $H\left(W^L\right)$ as follows.
\begin{small}
\begin{eqnarray}
H\left(W^{L+1}\right) &=& \frac{H\left(W^{L}\right)\times L - \log T}{L+1} \nonumber \\
\text{where} \,\,\, T &=& \left\{
\begin{aligned}
&\frac{L^L}{(L+1)^{L+1}}\cdot\frac{\left(n^{L+1}(v)\right)^{n^{L+1}(v)}}{\left(n^L(v)\right)^{n^L(v)}}, \,\, if\ v \in W^L\\
&\frac{L^L}{(L+1)^{L+1}}, \,\, if\ v \not\in W^L
\end{aligned}
\right.
\label{eq:incom_walk}
\end{eqnarray}
\end{small}
\end{theor}
%
\begin{proof}
From the definition of entropy of $W^L$ in Eq.~\ref{path_entropy}, we get:
\begin{footnotesize}
\begin{eqnarray}
\displaystyle H\left(W^L\right) &=& -\sum_{u \in W^L}\frac{n^L(u)}{L}\log\frac{n^L(u)}{L} \nonumber 
\end{eqnarray}
\end{footnotesize}
\begin{footnotesize}
\begin{eqnarray}
\displaystyle\implies  2^{-H\left(W^L\right)\cdot L} &=& 2^{\sum_{u \in W^L}n^L(u)\cdot\log\frac{n^L(u)}{L}}
=\frac{\prod_{u \in W^L}{\left(n^L(u)\right)}^{n^L(u)}}{L^L} \nonumber
\end{eqnarray}
\end{footnotesize}
Assuming that $v$ is in $W^L$,
\begin{footnotesize}
\begin{eqnarray}
\displaystyle L^L \cdot 2^{-H\left(W^L\right)\cdot L}&=&\left({n^{L}(v)}\right)^{n^{L}(v)}\cdot\prod_{u \in W^L\setminus\{v\}}\left({n^L(u)}\right)^{n^L(u)}
\label{eq:WL}
\end{eqnarray}
\end{footnotesize}
Analogously, for $W^{L+1}$ which is extended from $W^L$ by appending $v$, and assuming that $v$ is in $W^L$, we derive:
\begin{footnotesize}
\begin{eqnarray}
\label{eq:WL1}
\displaystyle (L+1)^{L+1} \cdot 2^{-H\left(W^{L+1}\right)\cdot (L+1)}&=&\left({n^{L+1}(v)}\right)^{n^{L+1}(v)}\cdot\prod_{u \in W^{L}\setminus\{v\}}\left({n^{L}(u)}\right)^{n^{L}(u)} \nonumber \\
\end{eqnarray}
\end{footnotesize}
Comparing Eq.~\ref{eq:WL} and \ref{eq:WL1}, we get:
\begin{footnotesize}
\begin{eqnarray}
\displaystyle 2^{H\left(W^L\right)\cdot L-H\left(W^{L+1}\right)\cdot (L+1)} &=& \frac{L^L}{(L+1)^{L+1}}\cdot\frac{\left(n^{L+1}(v)\right)^{n^{L+1}(v)}}{\left(n^L(v)\right)^{n^L(v)}} = T
\end{eqnarray}
\end{footnotesize}
The proof when $v \not \in W^L$ can be derived similarly.
\end{proof}
%
\underline{Incremental computing for walk termination.}
To terminate a random walk, {\sf HuGE} computes and verifies the
linear relation between information entropy $H$ and walk length $L$ at every step of the walk.
From Eq. \ref{path_corr}, $R(H,L)$ can be expressed as:
\begin{small}
\begin{equation}
\label{correlation}
\displaystyle R(H,L) = \frac{E(HL)-E(H)E(L)}{\sqrt{{(E(H^2)-E(H)^2)}{(E(L^2)-(E(L)^2)}}}
\end{equation}
\end{small}
The mean function $E($ $)$ can be computed incrementally.
\begin{footnotesize}
\begin{eqnarray}
\label{correlation_inc}
&& \displaystyle E_p(X)=\frac{1}{p}\sum\limits_{i=1}^{p}X_i = \left(\frac{p-1}{p}\right)E_{p-1}(X)+\frac{X_p}{p} \\
&& \displaystyle E_p(XY)=\frac{(p-1)^2E_{p-1}(XY)+(p-1)[X_pE_{p-1}(Y)+Y_pE_{p-1}(X)]+X_pY_p}{p^2} \nonumber
\end{eqnarray}
\end{footnotesize}

\underline{Message size reduction.}
Due to incremental computation, instead of full-path information,
only constant-length messages need to be sent across computing machines: [$walker\_id$, $steps$, $node\_id$, $H$, $L$, $E(H)$, $E(L)$, $E(HL)$, $E(H^2)$, $E(L^2)$].
\begin{figure}
  \centering
  \includegraphics[width= 3.2 in]{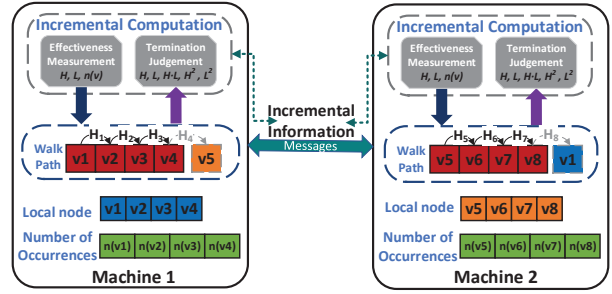}
  \caption{\small Incremental computing for information-centric random walk}
  \label{InCoM_computation}
\end{figure}
\begin{exam}
\label{ex:incom}
Figure~\ref{InCoM_computation} exhibits incremental information-centric computing ({\sf InCoM}).
The graph is partitioned into two machines: $\{v_1, v_2, v_3, v_4\}$ in $M_1$ and
$\{v_5, v_6, v_7, v_8\}$ in $M_2$.
A local frequency list is maintained for each ongoing walk at every machine,
having \# occurrences information only about nodes that are local to a machine.
When a local node is added to the walk, \# occurrences for this node
is updated in this local list. Using the local frequency list, {\sf DistGER}
incrementally computes information entropy $H$, as well as the linear relation $R$ between
$H$ and the current path length $L$ to decide on walk termination.
If the next accepted node is not at the current machine, the walker only needs to
carry the necessary incremental (constant-length) information as a message including
$walker\_id$, $steps$, $node\_id$, $H$, $L$, $E(H)$, $E(L)$, $E(HL)$, $E(H^2)$, and $E(L^2)$
to the other machine, and then generate $H$ and $R$ in that machine. Given an 8 bytes
space to store a variable, since the messages in baseline {\sf HuGE-D} carry the full-path information (i.e., [$walk\_id$, $steps$, $node\_id$, $path\_info$]), {\sf HuGE-D} needs $24 + 8L$ bytes per message, where $L$ is the walk length, while {\sf DistGER} only requires the constant size of 80 bytes. If the maximum path length is 80 (commonly used), one message in {\sf DistGER} is up to 8.3$\times$ smaller than that in {\sf HuGE-D}.
\end{exam}

\underline{Complexity analysis.}
The pseudocode of {\sf InCoM}
remains similar to that in Algorithm~\ref{HuGE-D_walk}, except that $H$ and $R$ are
computed incrementally in Line 6, via Eq. \ref{eq:incom_walk} and \ref{correlation_inc},
and we require only $\bigO(1)$ time at each step of the walk. Furthermore, when a walk crosses a machine,
it sends 
a constant-length message to another machine. Following similar analysis as {\sf HuGE-D},
the time spent for {\sf InCoM} 
is:
$\bigO(r'\cdot L'\cdot |V|/P + N\cdot M(1)/B)$,
where $P$, $B$, and $N$ are \# processors, network bandwidth,
and \# cross-machine messages, respectively.
\subsection {Multi-Proximity-aware Streaming Partitioning}
\label{sec:partition}
Graph partitioning aims at balancing workloads across machines
in a distributed setting, while also reducing cross-machine
communications.
Balanced graph partitioning with the minimum edge-cut is an \NP-hard
problem \cite{BulucMSS016}. Instead, our system, {\sf DistGER} develops
multi-proximity-aware, streaming, parallel graph partitioning ({\sf MPGP}): By leveraging first
and second-order proximity measures,
we select a good-quality partitioning to ensure that each walker stays in a local computing
machine as much as possible, thereby reducing the number of cross-machine communications.
The partitioning is conducted in a node streaming manner, hence it scales to larger graphs
\cite{streaming_model_SIGMOD_2019, streaming_model_VLDB_2018}.
We also ensure workload balancing among the computing servers.

\underline{Partitioning method.}
Given a set of partially computed partitions, $P_1, P_2,..., P_m$, where $m$ denotes
\# machines, an un-partitioned node $v$ is placed in one of these partitions based on the following objective:
\begin{small}
\begin{eqnarray}
&\displaystyle \argmax_{i \in \{1...m\}}\left(PS_1(v,P_i) + PS_2(v,P_i)\right)\times \tau(P_i) \\
& \displaystyle \text{where} \, \, \, \tau(P_i) = 1-\frac{|P_i|}{\gamma \times (\sum \limits_{i=1}^{m} |P_i|)/m} \noindent
\label{Mpad_constrait}
\end{eqnarray}
\end{small}
$PS_1(v, P_i)=|\{u\in N(v)\cap P_i\}|$ and $PS_2(v, P_i)=\sum_{u\in P_i}|\{N(v)\cap N(u)\}|$
represent the first- and second-order proximity scores, respectively.
For edge weight $w(v,u)$, $PS_1(v, P_i)=
\sum_{u\in N(v)\cap P_i}w(v,u)$ and $PS_2(v, P_i)=\sum_{u\in P_i}|\{N(v)\cap N(u)\}| \cdot w(v, u)$.
Intuitively,
$PS_1$ denotes \# neighbors of an unpartitioned node in the target partition, thus a higher value of $PS_1$ implies that the unpartitioned node should have a higher chance to be assigned to the target partition.
Since $PS_2$ is defined by \# common neighbors, which are widely used to measure the similarity of node-pairs during the random walk, a higher value of $PS_2$ is also in line with the characteristics of random walk.
Higher first- and second-order proximity scores increase the chance that a random walker
stays in a local computing machine, thereby reducing cross-machine communication.
{\sf MPGP} also introduces a dynamic load-balancing term $\tau(P_i)$,
where $|P_i|$ denotes the current number of nodes in $P_i$,
hence it is updated after every un-partitioned node's assignment;
and $\gamma$ is a slack parameter that allows deviation from the exact load balancing.
Setting $\gamma=1$ ensures strict load balancing but hampers partition quality.
Meanwhile, setting a larger $\gamma$ relaxes the load-balancing constraint, and creates a skewed partitioning.

{\sf MPGP} differs from some of the existing node streaming-based graph partition schemes, e.g., {\sf LDG} \cite{LDG_2012} and {\sf FENNEL} \cite{FENNEL2014}
in several ways. {\bf First,}
{\sf LDG} and {\sf FENNEL}
set a maximum size for each partition in advance based on the total number of nodes and tend to assign
nodes to a partition until the maximum possible size is reached for that partition.
In contrast, we ensure good load balancing across partitions at all times during the partitioning phase.
{\bf Second,} we find that {\sf LDG} and {\sf FENNEL} cannot partition larger graphs in a reasonable time, for example,
they consume more than one day to partition the Youtube \cite{Flickr_Youtube_Graph} graph with 1M nodes and 3M edges,
while our proposed {\sf MPGP} requires only a few tens of seconds (\S \ref{sec:experiments}), which is due to our optimization methods
discussed below.

\underline{Optimizations and parallelization.}
{\bf First,} to measure first-order proximity scores,
we apply the {\sf Galloping} algorithm \cite{Galloping_2000} that can speed-up intersection computation
between two unequal-size sets. During our streaming graph partitioning,
the partition size gradually increases, hence the {\sf Galloping} algorithm is quite effective.
{\bf Second,} during a second-order proximity score computation, i.e., $PS_2(v, P_i)=\sum_{u\in P_i}|\{N(v)\cap N(u)|$,
we only consider those nodes $u\in P_i$ whose contributions to first-order proximity score are non-zero,
i.e., $u\in N(v)$. This is because if $u$ is not a neighbor of $v$, the random walk cannot reach $u$.
{\bf Third,} nodes streaming order could impact the partitioning time and effectiveness.
We compare a number of streaming orders \cite{LDG_2012,FENNEL2014}, e.g., random, BFS, DFS, and their variations.
Based on empirical results (\S \ref{sec:experiments}), we recommend {\sf DFS+degree}-based streaming for sequential {\sf MPGP}:
Among the un-explored neighbors of a node during DFS, we select the one having the highest degree.
This strategy improves the efficiency of the {\sf Galloping} algorithm.
{\bf Fourth}, with large-scale graphs (e.g., Twitter), sequential MPGP,
still requires considerable time to partition. Thus, we implement a simple parallelization scheme {\sf parallel MPGP} ({\sf MPGP-P}),
as follows:
We divide the stream into several segments and independently partition the nodes of each segment in parallel
via {\sf MPGP}; finally, we combine the partitioning results of all segments.
Based on our empirical results (\S \ref{sec:experiments}), we recommend {\sf BFS+Degree} for {\sf MPGP-P} since
it reduces the partition time greatly and the random walk time on a partitioned graph is comparable to that obtained from the sequential version of MPGP.

\underline{Complexity analysis.}
The running time of {\sf MPGP} is dominated by first- and second-order proximity scores computation for each node,
which are computed in parallel for all $m$ partitions. 
For a first-order proximity score computation via
the {\sf Galloping} algorithm,
let the smaller set size be $S_1$.
In the early stages of partitioning, the larger set constitutes the neighbors of an
un-partitioned node $v$, then the time complexity is $\bigO\left(S_1 \cdot \log|N(v)|\right)$; while at later stages, the larger set is the partition,
thus it takes $\bigO\left(S_1 \cdot \log(|V|/m)\right)$ time.
For the second-order proximity score computing of $v$,
let $S_2$ be the intersection set size generated by the first-order proximity scores computing of $v$.
As the number of common neighbors for each edge $(u,v)$ is processed in parallel by the Galloping algorithm,
the second-order proximity score computing requires $\bigO(\frac{S_2}{T}\cdot|N(v)|\cdot\log N_{max}^v)$ time,
with $T$ threads, where $N_{max}^v=\max_{u\in N(v)}|N(u)|$.

%

\underline{Benefits over {\sf KinghtKing}'s graph partitioning:}
Compared to only load-balancing based partition in {\sf KnightKing} that introduces high communication cost,
our {\sf MPGP} aims at both load-balancing and reducing cross-machine communications.
On average, {\sf MPGP} reduces 45\% cross-machine communications for all graphs that we use in our experiments, resulting about 38.9\% efficiency improvement (\S \ref{sec:experiments}).
%
%
\section{Distributed Embedding Learning}
\label{sec:learning}
Our {\sf learner} in {\sf DistGER} supports distributed learning of node embeddings using the random walks
generated by the {\sf sampler} (Figure~\ref{HuGER_framework}).
We first discuss the shortcomings of state-of-the-art methods
on distributed {\sf Skip-Gram} with negative sampling and provide an overview of our solution (\S \ref{sec:skipgram_challenges}), then
elaborate our three novel improvements (\S \ref{sec:opt_skipgram}), and finally summarize our overall approach (\S \ref{sec:everything}).
\subsection{Challenges and Overview of Our Solution}
\label{sec:skipgram_challenges}
{\sf Skip-Gram} uses the stochastic gradient descent ({\sf SGD}) \cite{SGD_1951}.
Parameters update from one iteration and gradient computation in the next iteration may touch the same node
embedding, making it inherently sequential.
The original {\sf word2vec} leverages {\sf Hogwild} \cite{2011HOGWILD}
for parallelizing {\sf SGD}. It asynchronously processes different node pairs in parallel and ignores any conflicts
between model updates on different threads.
However, there are shortcomings in this approach \cite{Pword2vec_2019,pSGNSCC_2017}.
(1) Since multiple threads can update the same cache lines, updated data needs to
be communicated between threads to ensure cache-coherency, introducing cache lines ping-ponging
across multiple cores, which results in high access latency.
(2) {\sf Skip-Gram} with negative sampling randomly selects a set of nodes as
negative samples for each context node in a context window.
Thus, there is a certain locality in model updates for the same target node, but this feature has not been exploited
in the original scheme.
The randomly generated set of negative samples also introduces
random access for parameter updates in each iteration, degrading performance.

As shown in Figure~\ref{SG_pw2v_dsgl}(a), a walk denoted as $W_1$ is assigned to a
thread, and there is a sliding window $w_1$ containing context nodes $\{v_1, v_2, v_4, v_5\}$ and the target node $v_3$.
The {\sf Skip-Gram} computes the dot-products of word vectors $\varphi_{in}({v_i})$ for a given word $v_i \in \{v_1, v_2, v_4, v_5\}$
and $\varphi_{out}({v_3})$, as well as for a set of $K$ negative samples, $\varphi_{out}({v_k})$ ($v_k \in V$).
Notice that $\varphi_{out}(v_3)$ will be computed four times with four different context words and sets of
negative samples, and these dot-products are level-1 {\sf BLAS} operations \cite{BLAS_2002}, which are limited by memory bandwidth.
Some state-of-the-art work, e.g., {\sf Pword2vec} \cite{Pword2vec_2019} shares
negative samples with all other context nodes in each window (Figure \ref{SG_pw2v_dsgl}(b)), and thus converts
level-1 BLAS vector-based computations into level-3 BLAS matrix-based computations to efficiently utilize computational resources.
However, such matrix sizes are relatively small and still have a gap to reach the peak CPU throughput.
Besides, this way of sharing negative samples cannot significantly reduce the ping-ponging of cache lines across multiple cores.
To improve CPU throughput, {\sf pSGNScc} \cite{pSGNSCC_2017} combines context nodes of the negative sample from
another window into the current window; and together with the current context nodes, it generates larger matrix batches
(Figure \ref{SG_pw2v_dsgl}(c)). Nevertheless, {\sf pSGNScc} needs to maintain a pre-generated inverted index table to find related windows, resulting in additional space and lookup overheads.
%
\begin{figure}
  \centering
  \includegraphics[width= 3.2 in]{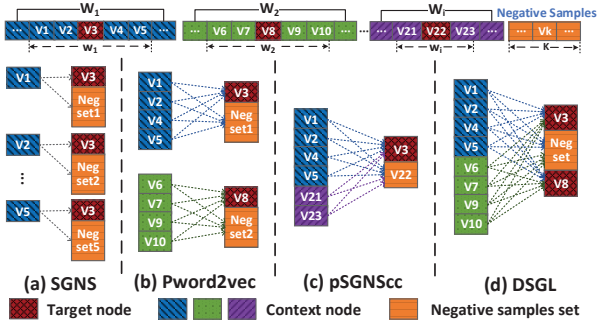}
  \caption{\small Schematic diagram of (a) {\sf Skip-Gram} with negative samples ({\sf SGNS}), (b) {\sf Pword2vec} \cite{Pword2vec_2019}, (c) {\sf pSGNScc} \cite{pSGNSCC_2017}, and (d) {\sf DSGL} (our method)}
  \label{SG_pw2v_dsgl}
\end{figure}

To address the above limitations, we propose a distributed {\sf Skip-Gram} learning model, named {\sf DSGL}
(Figure~\ref{SG_pw2v_dsgl}(d)).
{\sf DSGL} leverages global-matrices and local-buffers for the vectors of context nodes
and target/negative sample nodes during model updates to improve the locality and reduce the ping-ponging of cache lines across multiple cores.
We then propose a multi-windows shared-negative samples computation mechanism to fully exploit the CPU throughput.
Last but not least, {\sf DSGL} uses a hotness-block based synchronization mechanism to synchronize the node vectors in
a distributed setting.
While we design them considering information-oriented random walks, they are generic
and have the potential to improve any random walk and distributed {\sf Skip-Gram} based
graph embedding model (\S \ref{sec:experiments}).
\subsection{Proposed Improvements: DSGL}
\label{sec:opt_skipgram}
\underline{Improvement-I: Global-matrix and local-buffer.}
As reasoned above, the parallel {\sf Skip-Gram} with negative sampling
suffers from poor locality and ping-ponging of cache lines across multiple cores.
Since the model maintains two matrices to update the parameters through
forward and backward propagations during training, where the input
matrix $\varphi_{in}$ and output matrix $\varphi_{out}$ store the vectors of
context nodes and target/negative sample nodes, respectively -- how to construct
these two matrices are critical to the locality of data access.

Since most of the real-world graphs follow a power-law degree distribution \cite{BA99,ABA03},
we find that the generated corpus sampled from those graphs also has this feature --
a few nodes occupy most of the corpus, which are updated frequently during training.
In light of this observation, we construct $\varphi_{in}$ and $\varphi_{out}$ as {\em global matrices}
in descending order of node frequencies from the generated corpus.
It ensures that the vectors of high-frequency nodes stay in the cache lines
as much as possible. Recall that node frequencies in the corpus were already computed in the random
walk phase (\S \ref{sec:randGembedding}).

In addition, to avoid cache lines ping-ponging, {\sf DSGL} uses {\em local buffers} to update the context and
negative sample nodes for each thread, thus the node vectors are
first updated in the local buffers within a lifetime (i.e., when a thread processes a walk during training)
and are then synchronized to the global matrix. 
Precisely, since a sliding window $w$
always shifts the boundary and the target node by one node (Eq.~\ref{skim_gram_eq}), each node in
a walk will appear as a context node in up to $2w$ sliding windows
and as a target node once, thus a context node can be reused up to $2w+1$
times in the lifetime. It is necessary to consider this
temporal locality and build a {\em local context buffer} for each node vector
in a walk.
Meanwhile, on-chip caches where randomly-generated
negative sample nodes are located, have a higher chance to be
updated by multiple threads, and updating their vectors to the
global $\varphi_{out}$ requires low-level memory accesses.
To alleviate this, {\sf DSGL} also constructs a {\em local negative buffer}
for vectors of negative samples during one lifetime.
It randomly selects $K$ (negative sample set size) $\times$ $L$ (walk length, i.e., total steps) negative samples into the negative buffer from $\varphi_{out}$. Thus, {\sf DSGL} can use a different negative sample set with the same size ($K$) at each step.
For the target node, due to its lower re-usability compared to context nodes,
it is only updated once in global $\varphi_{out}$ in the lifetime;
hence, we do not create a buffer for it.
The constructed local buffers
require small space,
since the sizes of buffers are related to the walk length.
The length of the information-centric random walk is much
smaller than that of traditional random walks (\S \ref{sec:randGembedding}).
\begin{figure}[tb!]
  \centering
  \includegraphics[width= 3.2 in]{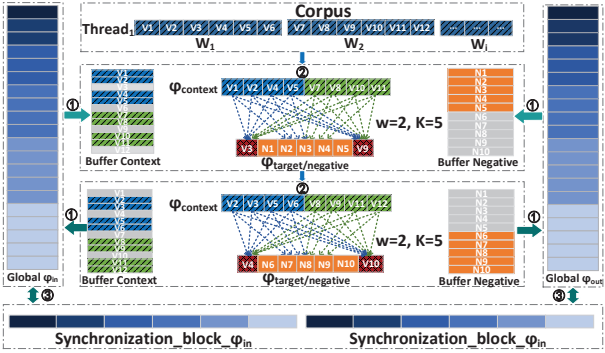}
  \caption{\small Workflow of our {\sf DSGL}: distributed Skip-Gram learning}
  \label{DSGL_framework}
\end{figure}

\underline{Improvement-II: Multi-windows shared-negative samples.}
CPU \par
\noindent throughput of {\sf Skip-Gram} model can be significantly improved by
converting vector-based computations into matrix-based computations
\cite{floating_2015};
however, the batch matrix sizes are relatively small for existing methods \cite{Pword2vec_2019,pSGNSCC_2017}
and cannot fully utilize CPU resources. Based on this consideration, we design a multi-windows
shared-negative samples mechanism: To increase the batch matrix sizes, we
batch process the vector updates of context windows from multiple ($\geq$ 2) walks allocated to the same thread.

Consider Figure~\ref{SG_pw2v_dsgl}(d),
two walks $W_1$ and $W_2$ are assigned to the same thread, where $w_1$ and $w_2$ are two context windows in $W_1$ and $W_2$;
$v_3$ and $v_8$ are target nodes in $w_1$ and $w_2$, respectively.
{\sf DSGL} simultaneously batch updates the node vectors
in two context windows $w_1$ and $w_2$ based on level-3 {\sf BLAS} matrix-based computations,
the set of negative samples is shared across the batch context nodes,
$v_3$ and $v_8$ are used as additional negative samples for $w_2$ and $w_1$, respectively.
After the lifetime of $W_1$ and $W_2$, the updated parameters of all context nodes and target/negative samples are
written back to $\varphi_{in}$ and $\varphi_{out}$, respectively.
Assuming the set of negative samples $K=5$, the batch matrix sizes of one iteration for {\sf Pword2vec}  \cite{Pword2vec_2019} (Figure~\ref{SG_pw2v_dsgl}(b))
is $4\times6$,
while that of {\sf DSGL} is extended to
$8\times7$.
This utilizes higher CPU throughput
and accelerates training without sacrificing accuracy(\S~\ref{sec:experiments}).
In practice, the number ($\geq$ 2) of multi-windows can be flexibly set according to available hardware resources.

\underline{Improvement-III: Hotness-block based synchronization.} In a distributed setting, the updated node vectors need to be synchronized with each computing machine.
Assuming that the generated corpus is partitioned into $m$ machines, each machine independently
processes the local corpus and periodically synchronizes the local parameters with other $m - 1$ machines.
For a full model synchronization across computing machines, the communication cost is $\bigO(|V|\cdot d \cdot m)$,
where $|V|$ is the total number of nodes and $d$ denotes the dimension of vectors, e.g.,
for 100 million nodes with 128 dimensions, the full model synchronization needs $\approx102.4$ billion messages across 4 machines --
it will be difficult to meet the efficiency requirement.

Notice that the node occurrence counts in generated corpus
also follows a power-law distribution, implying that high-frequency nodes have a higher probability of being accessed
and updated during training.
Following this, we propose a
hotness-block based synchronization mechanism in {\sf DSGL}. Instead of
full synchronization,
we only conduct more synchronization for hot nodes than that for
low-frequency nodes. Since our global matrices are constructed based on the descending order of node frequencies in the generated corpus (Improvement-I),
it provides a favorable condition to achieve hotness-block based synchronization.
The global matrices are partitioned into several blocks (i.e., hotness-blocks)
based on the same frequency of nodes in the corpus, and are denoted as $B(i)$, $0<i\leq ocn_{max}$,
where $ocn_{max}$ is the largest number of occurrences of any node in the corpus.
We randomly sample one node from each hotness-block; for all these sampled nodes,
we synchronize their vectors across all computing machines during one synchronization period.
Due to the node frequency skewness in the corpus, for each node in $B(i)$, the probability of it
being sampled for synchronization is inversely proportional to $|B(i)|$. Thus, we
ensure that the hot nodes are synchronized more during the entire training procedure, while the low-frequency
nodes would have relatively less synchronization. Compared to the full model synchronization mechanism,
our synchronization cost is $\bigO(ocn_{max}\cdot d \cdot m)$, where  $ocn_{max}<<$ $|V|$,
indicating that it can significantly reduce the load on network, while keeping the parameters on each computing machine updated aptly.
\subsection{Putting Everything Together}
\label{sec:everything}
First, {\sf DSGL} constructs two global matrices $\varphi_{in}$ and $\varphi_{out}$ in
descending order of node frequencies from the generated corpus.
It leverages a pipelining construction during the random walk phase,
one thread is responsible for counting the frequency of each node on its local walk;
at the end of all random walks, these counts from computing machines are aggregated
to construct global $\varphi_{in}$ and $\varphi_{out}$.
Next, {\sf DSGL} uses a multi-windows shared-negative samples
computation for each thread.
As shown in Figure \ref{DSGL_framework},
two walks $W_1$ and $W_2$ are assigned to $Thread_1$,
suppose \# multi-windows = 2 and
negative sample set size $K$ = 5.
For maintaining access locality and reducing cache lines ping-ponging,
two local buffers per thread are constructed for the context (blue and green blocks) and negative sample (orange blocks) nodes. The target nodes are denoted as red blocks.
Initially, the two buffers will respectively load the vectors associated
with nodes from $\varphi_{in}$ and $\varphi_{out}$,
then $Thread_1$ proceeds over $W_1$ and $W_2$ with matrix-matrix multiplications,
and the updated vectors are written to the buffers at each step.
After the lifetime of $W_1$ and $W_2$, the latest node vectors in buffers are
written back to global $\varphi_{in}$ and $\varphi_{out}$.
{\sf DSGL} periodically synchronizes the updated vectors across multiple
computing machines by hotness-block based synchronization.

\section{Related Work}
\label{sec:related}
%
\spara{Random walk-based graph engines. }
{\sf Graphwalker} \cite{GraphWalker_2020} focuses on improving I/O efficiency to support
scalable random walks on a single machine. {\sf FlashMob} \cite{FlashMob_2021} optimizes
cache utilization to make memory accesses more sequential 
for faster random walks.
{\sf ThunderRW} \cite{ThunderRW_2021} addresses irregular memory access patterns of random walks.
{\sf CraSorw} \cite{CraSorw_2022} proposes an I/O-efficient disk-based system to
support scalable second-order random walks.
The GPU-based system, {\sf NextDoor} \cite{nextdoor_2021} leverages intelligent scheduling and caching strategies to mitigate
irregular memory access patterns in random walks.
Different from single-machine random walk systems,
{\sf KnightKing} \cite{KnighKing_2019} 
provides a general-purpose distributed random walk engine.
However, it ignores information-effectiveness during random walks and is limited by ineffective
graph partition 
(\S~\ref{sec:dRand}).
%

\spara{Graph embedding algorithms.} 
Sequential graph embedding techniques \cite{CaiZC18} fall into three categories.
Matrix factorization-based algorithms \cite{ProNE_2019,NetSMF_2019,NRP_2020,QiuDMLWT18,WangCWP0Y17}
construct feature representations based on the adjacency or Laplacian matrix, and involve spectral techniques \cite{spectral_2001}.
Graph neural networks (GNNs)-based approaches \cite{TuCWY018,WangC016,GraphSAGE_2017,Graph_attention_2018,Graphgan_2018}
focus on generalizing graph spectra into
semi-supervised or supervised graph learning. 
Both techniques incur high computational overhead and DRAM dependencies, limiting their scalability to large graphs.
Random-walk methods \cite{DeepWalk_2014,node2vec_2016,HuGE_2021,Line_2015,HuGE+_2022} transform a graph into a set of random walks through sampling and then adopt {\sf Skip-Gram} to generate embeddings.
They are 
more flexible, parallel-friendly, and scale to larger graphs \cite{FlashMob_2021}.
{\sf HuGE+} \cite{HuGE+_2022} is a recent extension of {\sf HuGE} \cite{HuGE_2021}, which considers the information content of a node during the next-hop node selection to improve the downstream task accuracy. It uses the same {\sf HuGE} information-oriented method to determine the walk length and number of walks per node (\S \ref{sec:randGembedding}), and hence the efficiency is similar to that of {\sf HuGE}.


\spara{Graph embedding systems and distributed embedding.} To address efficiency challenges with large graphs, 
recently proposed {\sf GraphVite} \cite{GraphVite_2019} and
Tencent's graph embedding system \cite{Tencent_GE_2020} follows sampling-based techniques on a
CPU-GPU hybrid architecture, 
simultaneously performing graph random
walks on CPUs and embedding training on GPUs. {\sf Marius} \cite{marius_2021} optimizes
data movements between CPU and GPU on a single machine for large-scale knowledge graphs embedding.
{\sf Seastar} \cite{seastar_2021} develops
a novel GNN training framework on GPUs with a vertex-centric \cite{luo2023multi} programming model. 
We deployed {\sf DistGER} on GPU, but it does not provide a significant improvement, especially for large-scale graphs. Similar to the above works, computing gaps between CPUs and GPUs and the limited memory of GPUs still plague the efficiency of graph embedding.

Other approaches attempt to scale graph embeddings from a distributed perspective.
{\sf HET-KG} \cite{dong2022het} 
is a distributed system for knowledge graph embedding. It introduces a cache embedding table to reduce communication overheads among machines.
{\sf AliGraph} \cite{AliGraph_2019} optimizes sampling operators for distributed GNN training and reduces network communication by caching nodes
on local machines. Amazon has released {\sf DistDGL} \cite{DistDGL_2020}, a distributed graph embedding framework for GNN model with mini-batch training
based on the {\sf Deep Graph Library} \cite{dgl_2019}. {\sf Pytorch-Biggraph} \cite{PBG_2019} leverages graph partitioning and parameter servers
to learn large graph embeddings on multiple CPUs in a distributed environment based on {\sf PyTorch}.
However, the efficiency and scalability of {\sf DistDGL} and {\sf Pytorch-BigGraph} are affected by parameter synchronization
(as demonstrated in \S\ref{sec:experiments}). {\sf ByteGNN} \cite{ZhengCCSWLCYZ22} is a recently proposed
distributed system for GNNs, with mini-batch training and two-level scheduling to improve parallelism and resource utilization,
and tailored graph partitioning for GNN workloads. 
Since {\sf ByteGNN} is not publicly available yet, we cannot compare it in our experiments.

There are also approaches attempting to address computational efficiency challenges with new hardware. {\sf Ginex} \cite{Ginex_2022} proposes an SSD-based GNN training system with feature vector cache optimizations that can support billion-edge graph datasets on a single machine. {\sf SmartSAGE} \cite{SMARSAGE_2022} develops software/hardware co-design based on an in-storage processing architecture for GNNs to overcome the memory capacity bottleneck of training. {\sf ReGNN} \cite{liu2022regnn} designs a ReRAM-based processing-in-memory architecture consisting of analog and digital modules to process GNN in different phases effectively. Although new hardware-based approaches open up opportunities for training larger datasets or providing more accelerations, their programming compatibility and prohibitive expensive still pose challenges.
\begin{table*}[t!]
\begin{minipage}{0.175\linewidth}
\centering
\captionof{table}{\small Datasets statistics \label{graph_datasets}}
\begin{tiny}
\begin{tabular}{c||c|c}
      {\bf Graph} & {\bf \#nodes} & {\bf \#edges} \\ \hline
      {\em FL} & 80\,513     & 5\,899\,882 \\
      {\em YT} & 1\,138\,499 & 2\,990\,443 \\
      {\em LJ} & 2\,238\,731 &14\,608\,137 \\
      {\em OR} & 3\,072\,441 & 117\,185\,083 \\
      {\em TW} & 41\,652\,230 & 1\,468\,365\,182 \\
\end{tabular}
\end{tiny}
\end{minipage}%
 \quad
 \begin{minipage}{.265\linewidth}
\centering
\tabcolsep=0.05cm

\captionof{table}{\small Avg. memory footprint (GB) of {\sf DistGER} and {\sf KnightKing} on each machine, where $\sigma$ is the standard deviation}
\label{Memory_usage}
\begin{tiny}
\newcommand{\tabincell}[2]{\begin{tabular}{@{}#1@{}}#2\end{tabular}}
  \begin{tabular}{c|cc|cc}
    { }&\multicolumn{2}{c|}{\bfseries{ Sampling}}&\multicolumn{2}{c}{\bfseries{Training}}\\
    \hline
    {\bf{Graph}} &{\sf KnightKing} &{\sf DistGER} &{\sf KnightKing} &{\sf DistGER} \\
    \hline
     {\em FL} & 0.66($\pm$0.06)	&{\bf 0.41($\pm$0.02)}	&1.31($\pm$0.17) 	&{\bf 0.86($\pm$0.06)} 	\\

     {\em YT} &4.11($\pm$0.55)	&{\bf 1.36($\pm$0.23)} 	&4.73($\pm$0.72) 	&{\bf 4.26($\pm$0.63)} \\

     {\em LJ} & 7.65($\pm$0.82)	&{\bf 1.95($\pm$0.16)}	&6.38($\pm$0.97) 	&{\bf 5.49($\pm$0.85)} 	\\

     {\em CO} &10.98($\pm$1.03)	&{\bf 3.27($\pm$0.79)} 	&8.52($\pm$1.01) 	&{\bf 6.86($\pm$0.69)} 	\\

     {\em TW} & out-of-memory	&{\bf 20.18($\pm$3.62)} 	&out-of-memory 	& {\bf 67.16($\pm$5.18)} 	\\
\end{tabular}
\end{tiny}

\end{minipage}
\quad
\begin{minipage}{.25\linewidth}
    \centering
    \includegraphics[width= 1.85in, height = 1.2in]{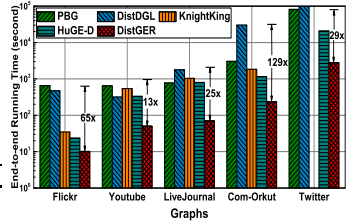}%
    \captionof{figure}
      {\small Efficiency: {\sf PBG} \cite{PBG_2019}, {\sf DistDGL} \cite{DistDGL_2020}, {\sf KnightKing} \cite{KnighKing_2019}, {\sf HuGE-D} (baseline), {\sf DistGER} (ours)
        \label{overall_performance}
      }
\end{minipage}
\quad
\begin{minipage}{.25\linewidth}
    \centering
    \includegraphics[width= 1.85in, height = 1.2in]{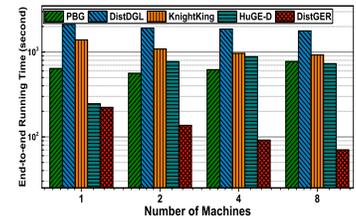}%
    \captionof{figure}
      {\small Scalability: {\sf PBG} \cite{PBG_2019}, {\sf DistDGL} \cite{DistDGL_2020}, {\sf KnightKing} \cite{KnighKing_2019}, {\sf HuGE-D} (baseline), {\sf DistGER} (ours)
        \label{Dist_scalability}
      }
\end{minipage}
\end{table*}

\section{Experimental Results}
\label{sec:experiments}
We evaluate the efficiency (\S \ref{sec:overall}) and scalability (\S \ref{sec:scalability}) of our proposed method, {\sf DistGER}
by comparing with {\sf HuGE-D} (baseline),
{\sf KnightKing} \cite{KnighKing_2019}, {\sf PyTorch-BigGraph} ({\sf PBG}) \cite{PBG_2019}, and {\sf Distributed DGL}
({\sf DistDGL}) \cite{DistDGL_2020}. We also compare the effectiveness (\S \ref{sec:effectiveness}) of generated embeddings
on link prediction.
Finally, we analyze efficiency due to individual
parts of {\sf DistGER} (\S \ref{sec:individual})
and the generality of {\sf DistGER} for other random walk-based embeddings (\S \ref{sec:generality}).
Our codes and datasets are at \cite{code}.
\subsection{Experimental Setup}
\label{sec:setup}
\spara{Environment.} We conduct experiments on a cluster of 8 machines with 2.60GHz Intel $^\circledR$ Xeon $^\circledR$ Gold 6240 CPU with 72 cores (hyper-threading)
in a dual-socket system, and each machine is equipped with 192GB DDR4 memory and connected by a 100Gbps network.
The machines run Ubuntu 16.04 with Linux kernel 4.15.0. We use GCC v9.4.0 for compiling {\sf DistGER}, {\sf KnightKing}, and {\sf HuGE-D},
and use Python v3.6.15 and torch v1.10.2 as the backend deep learning framework for {\sf Pytorch-BigGraph} and {\sf DistDGL}.

\spara{Datasets. } We employ five widely-used, real-world graphs
(Table~\ref{graph_datasets}): {\em Flickr} (FL) \cite{Flickr_Youtube_Graph},
{\em Youtube} (YT) \cite{Flickr_Youtube_Graph},
{\em LiveJournal} (LJ) \cite{BlogCatalog_Twitter_LiveJournal_Graph},
{\em Com-Orkut} (OR) \cite{com-orkut_2012}, and {\em Twitter} (TW) \cite{twitter_2010}.
The first two graphs are selected for multi-label node classification with distinct number of node labels 195 and 47, respectively, 
where labels in {\em Flickr} represent interest groups of users, and {\em Youtube}'s labels represent groups of viewers that enjoy common video genres. The last four graphs are used in link prediction. We also use synthetic graphs \cite{RMAT_2004} (up to 1 billion nodes, 10 billion edges) and a real-world {\em UK graph} \cite{BSVLTAG} (100M nodes, 3.7B edges) to assess the scalability of {\sf DistGER}.
Considering the default settings of popular random walk-based methods (e.g., Deepwalk, node2vec, HuGE), we use their undirected version.

\spara{Competitors.} We compare {\sf DistGER} against three state-of-the-art distributed graph embedding frameworks: the distributed random walk engine, {\sf KnightKing} {\scriptsize\url{https://github.com/KnightKingWalk/KnightKing}}
\cite{KnighKing_2019}; the distributed multi-relations based graph embedding system, {\sf PyTorch-BigGraph} ({\sf PBG})
{\scriptsize\url{https://github.com/facebookresearch/PyTorch-BigGraph}} \cite{PBG_2019} -- designed by Facebook; and
the distributed graph neural networks-based system, {\sf DistDGL} {\scriptsize\url{https://github.com/dmlc/dgl}}
\cite{DistDGL_2020} -- recently proposed by Amazon. We also implement {\sf HuGE-D}, a distributed version of
information-centric random walk-based graph embedding ({\sf HuGE} \cite{HuGE_2021}), on top of {\sf KnightKing},
served as our baseline. Since {\sf KnightKing} and {\sf HuGE-D} provide distributed support only for
random walk without that for embedding learning, we generate their node embeddings using
{\sf Pword2vec} {\scriptsize\url{https://github.com/IntelLabs/pWord2Vec}} \cite{Pword2vec_2019},
the most popular distributed {\sf Skip-Gram} system released by Intel.

\spara{Parameters.} For {\sf DistGER} and {\sf HuGE-D} random walks, we set
parameters $\mu$=0.995, $\delta$=0.001 based on information measurements (\S \ref{sec:preliminaries}),
while {\sf KnightKing} uses $L$=80 and $r$=10 that are routine configurations in the traditional
random walk-based graph embedding \cite{node2vec_2016, DeepWalk_2014, KnighKing_2019}. For {\sf DistGER}, {\sf KnightKing}, and {\sf HuGE-D} training,
we set the sliding window size $w$=10, number of negative samples $K$=5, and synchronization period=0.1 sec \cite{Pword2vec_2019},
and additionally, multi-windows number=2, $\gamma$=2 for {\sf DisrGER}.
For fair comparison across all systems, 
we set the embedding dimension $d$=128 that is commonly used \cite{HuGE_2021,node2vec_2016,DeepWalk_2014,Line_2015,Verse_2018,ProNE_2019},
and report the average running time for each epoch. For task effectiveness evaluations,
we find the best results from a grid search over learning rates from 0.001-0.1, \# epochs from 1-30,
and \# dimensions from 128-512.

\eat{
\begin{table}
\newcommand{\tabincell}[2]{\begin{tabular}{@{}#1@{}}#2\end{tabular}}
  \caption{\small Avg. memory footprint (GB) of {\sf DistGER} and {\sf KnightKing} on each machine, where $\sigma$ is the standard deviation.}
  \label{Memory_usage}
  \begin{center}
   \footnotesize
  \begin{tabular}{c|cc|cc}
    { }&\multicolumn{2}{c|}{\bfseries{ Sampling}}&\multicolumn{2}{c}{\bfseries{Training}}\\
    \hline
    {\bf{Graph}} &{\sf KnightKing} &{\sf DistGER} &{\sf KnightKing} &{\sf DistGER} \\
    \hline
     {\em Flickr} & 0.66($\pm$0.06)	&{\bf 0.41($\pm$0.02)}	&1.31($\pm$0.17) 	&{\bf 0.86($\pm$0.06)} 	\\

     {\em Youtube} &4.11($\pm$0.55)	&{\bf 1.36($\pm$0.23)} 	&4.73($\pm$0.72) 	&{\bf 4.26($\pm$0.63)} \\

     {\em LiveJournal} & 7.65($\pm$0.82)	&{\bf 1.95($\pm$0.16)}	&6.387($\pm$0.97) 	&{\bf 5.49($\pm$0.85)} 	\\

     {\em Com-Orkut} &10.98($\pm$1.03)	&{\bf 3.27($\pm$0.79)} 	&8.52($\pm$1.01) 	&{\bf 6.86($\pm$0.69)} 	\\

     {\em Twitter} & out-of-memory	&{\bf 37.1($\pm$5.28)} 	&out-of-memory 	& {\bf 79.5($\pm$7.27)} 	\\
\end{tabular}
\end{center}
\end{table}
}
\subsection{Efficiency and Memory Use w.r.t. Competitors}
\label{sec:overall}
%
We report the end-to-end running times of {\sf PBG}, {\sf DistDGL}, {\sf KnightKing}, {\sf HuGE-D}, and {\sf DistGER}
on five real-world graphs with the cluster of 8 machines in Figure~\ref{overall_performance}.
The reported end-to-end time includes the running time of partitioning, random walks (for random walk-based frameworks), and training procedures.
{\sf DistGER} significantly outperforms the competitors
on all these graphs, achieving a speedup ranging from 2.33$\times$ to 129$\times$. 
Recall that {\sf DistGER} is a similar type of system as {\sf KnightKing} and {\sf HuGE-D},
and our key improvements are discussed in \S \ref{sec:DistGER} and in \S \ref{sec:learning}.
Analogously, Figure~\ref{overall_performance} exhibits that our system, 
{\sf DistGER} achieves an average speedup of 9.25$\times$ and 6.56$\times$ compared with {\sf KnightKing} and {\sf HuGE-D}.
Notice that we fail to run {\sf KnightKing} on the largest {\em Twitter} dataset
because its routine random walk strategy requires more main memory space.
The advantage of information-centric random walk in {\sf HuGE} is almost wiped out in {\sf HuGE-D}
due to on-the-fly information measurements and the higher communication costs in a distributed setting.
The multi-relation-based {\sf PBG} leverages a parameter server to synchronize embeddings between clients,
resulting in more load on the communication network. As a result, {\sf PBG} is on average
26.22$\times$ slower than {\sf DistGER}. For graph neural network-based system {\sf DistDGL},
due to the long running time of graph sampling (e.g., taking 80\% of the overhead for the {\sf GraphSAGE}),
it is highly inefficient than other systems. For the billion-edge {\em Twitter} graph, it does not terminate in 1 day.
%
{Table ~\ref{Memory_usage}} shows {\sf DistGER}'s average memory footprint on each machine of the 8-machine cluster. 
Compared
to 
same type of system
{\sf KnightKing}, 
{\sf DistGER} requires less memory for sampling and training.

\subsection{Scalability w.r.t. Competitors}
\label{sec:scalability}
%
%
Figure~\ref{Dist_scalability} shows end-to-end running times of all competing
systems on the {\em LiveJournal} graph, as we increase \# machines
from 1 to 8 to evaluate scalability. {\sf DistGER} achieves better scalability than the other
four distributed systems.
{\sf PBG} leverages a parameter server and a shared network filesystem
to synchronize the parameters in the distributed model. 
When the number of machines increases, {\sf PBG} puts more load
on the communications network, resulting in poor scalability. Likewise, {\sf DistDGL}
is bounded by the synchronization overhead for gradient updates,
limiting its scalability.
Both {\sf KnightKing} and {\sf HuGE-D} suffer from higher communication costs during random walks,
due to their only workload-balancing partitioning scheme (\S \ref{sec:dRand}, \S \ref{sec:individual}).
Since {\sf HuGE-D} is implemented on top of {\sf KnigtKing},
it exhibits worse scalability due to high communication costs and on-the-fly information measurements in a distributed setting (\S \ref{sec:HUGED}).
In comparison, {\sf DistGER} incorporates multi-proximity-aware streaming graph partitioning and incremental computations
to reduce both communication and computation costs, it also employs hotness-block based parameters synchronization
during training to dramatically reduce the pressure on network bandwidth. Hence, {\sf DistGER} achieves better scalability than other systems.
Due to space limitations, we omit {\sf DistGER}'s scalability results on other graphs, which exhibit similar trends. On {\em Twitter}, the end-to-end running times {\sf DistGER} on 1, 2, 4, and 8 machines are 3090s, 1739s, 1197s, and 746s, respectively,
while on {\em Com-Orkut}, the results are 304s, 204s, 149s, and 89s, respectively. 
The results show a good linear relationship.

\begin{table}
\quad
\begin{minipage}{0.46\linewidth}
    \centering
    \includegraphics[width= 1.6 in]{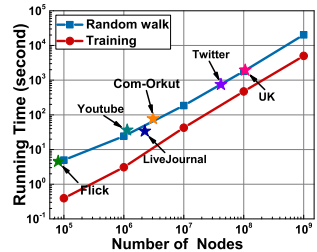}%
    \captionof{figure}
      {\small {Scalability of {\sf DistGER} on synthetic graphs, where Y-axis is in log-scale}}
        \label{Dist_scalability_data}
      
\end{minipage}\hfill
\quad
\begin{minipage}{.46\linewidth}
    \centering
    \includegraphics[width= 1.6 in]{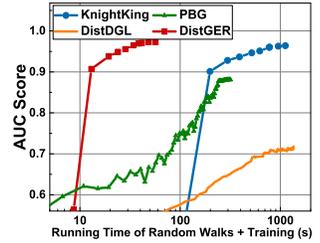}%
    \captionof{figure}
      {\small {The influence of running time on embedding quality for {\sf DistGER} and competitors}}
        \label{Dist_time_auc}
\end{minipage}
\end{table}

To further assess the scalability of {\sf DistGER}, we generate synthetic graphs \cite{RMAT_2004} with a fixed node degree of 10 and the number of nodes from $10^5$ to $10^9$. Figure~\ref{Dist_scalability_data} presents the running times for random walks and training on these synthetic graphs using a cluster of 8 machines, suggesting that the running time increases linearly with the size of a graph, and {\sf DistGER} has the capability to handle even billion-node graphs. Moreover, the running times for six real-world graphs (including the {\em UK graph} with $|E|=3.7B$, $|V|=100M$, for which the competing systems do not terminate in 1 day or crash due to hardware and memory limitation) are inserted into the plot, which is consistent with the trend on synthetic data.

\subsection{Effectiveness w.r.t. Competitors}
\label{sec:effectiveness}
\spara{Link prediction.} To perform link prediction on a given graph $G$, following \cite{HuGE_2021,node2vec_2016,Verse_2018,NRP_2020},
we first uniformly at random remove 50\% edges as positive test edges, and the rest are used as positive training edges.
We also provide negative training and test edges by considering those node pairs between which no edge exists in $G$.
We ensure that the positive and negative set sizes are similar. 
The link prediction is conducted as a classification task
based on the similarity of $u$ and $v$, i.e., $\varphi(u)\cdot\varphi(v)$.
The effectiveness of link prediction is measured via the $AUC$ (Area Under Curve) score \cite{AUC_kdd} -- the higher the better.
We repeat this procedure 50 times to offset the randomness of edge removal and report the average $AUC$ in
Table~\ref{AUC_results}.
{\sf DistGER} outperforms all competitors on these graphs, except for {\sf PBG} on {\em Com-Orkut}, where {\sf DistGER} ranks second.
On average, {\sf DistGER} has an 11.7\% higher $AUC$ score compared with the other three systems, thanks to our
information-centric random walks. {\sf PBG} is the best on {\em Com-Orkut} because this graph is much denser
and is friendly to the multi-relationship-based model in {\sf PBG}.
Figure~\ref{Dist_time_auc} exhibits accuracy-efficiency tradeoffs of {\sf DistGER} and competitors, i.e., their $AUC$ convergence curves w.r.t. increasing running times of random walks and training, over {\em LiveJournal}, further indicating
that {\sf DistGER} has better efficiency and effectiveness than the competitors.
%
\begin{table}[h!]
\newcommand{\tabincell}[2]{\begin{tabular}{@{}#1@{}}#2\end{tabular}}
  \caption{\small $AUC$ scores of {\sf DistGER} and competitors for link prediction}
  \label{AUC_results}
  \begin{center}
  \footnotesize
  \begin{tabular}{cccccc}
    {Method}&\tabincell{c}{Youtube}&{LiveJournal}&\tabincell{c}{Com-Orkut}&{ Twitter}\\
    \hline
    {\sf PBG}        & 0.753           &0.882            &\bfseries{0.955} &0.912\\

    {\sf DistDGL}    &0.894            &0.718            &0.815            & running time $>$ 1 day \\

    {\sf KnightKing} &0.904            &0.963            & $0.918$         & out-of-memory\\

    {\sf DistGER}    &\bfseries{0.966} &\bfseries{0.976} &0.921            &\bfseries{0.919}\\
\end{tabular}
\end{center}
\end{table}

\spara{Multi-label node classification.}
This task predicts one or more labels for each graph node and has applications in 
text categorization \cite{zhang2006multilabel} and bioinformatics \cite{zhang2018ontological}.
We use embedding vectors and a one-vs-rest logistic regression classifier
with L2 regularization \cite{MLC_LIBLINEAR_2008}, 
then evaluate the effectiveness by micro-averaged F1 ($Micro-F1$) and macro-averaged F1 ($Macro-F1$) \cite{WangC016}
scores, where $Micro-F1$ gives equal weight to each test instance and $Macro-F1$ assigns equal weight to each label category \cite{keikha2018community}.
Following \cite{HuGE_2021,node2vec_2016,DeepWalk_2014,Line_2015,Verse_2018},
we select 10\% to 90\% training data ratio on {\em Flickr}, and 1\% to 9\% training ratio on {\em Youtube}.
We report the averaged $Macro-F1$ and $Micro-F1$ scores from 50 trials in Figure~\ref{Dist_MLC_mac_mic_F1}.
We find that {\sf DistGER} has better $Macro-F1$ and $Micro-F1$ scores
than existing frameworks, 
gaining 9.2\% and 3.3\% average improvements, respectively, due to its more effective information-centric random walks.

\begin{figure}[h!]
  \centering
  \includegraphics[width= 3.45 in]{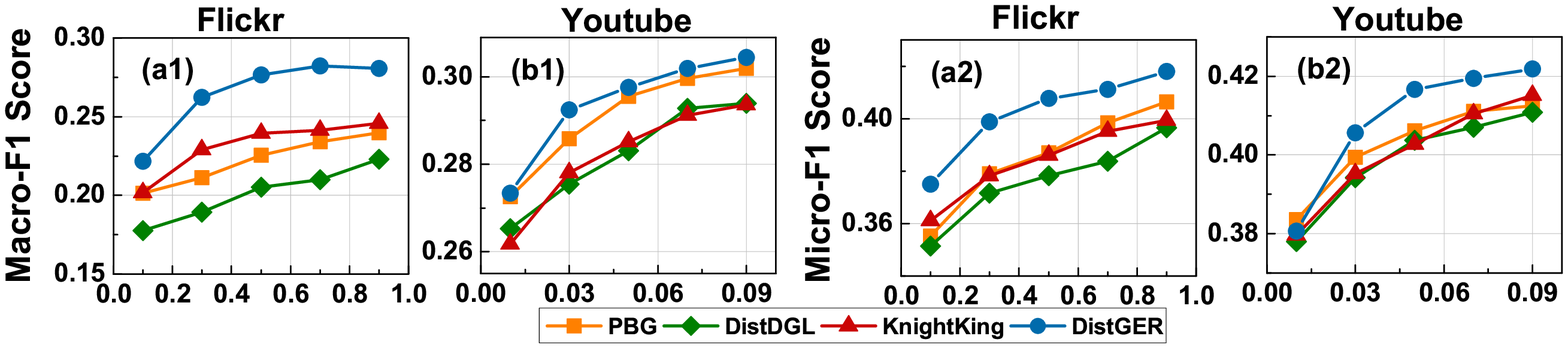}
  \caption{\small $Macro-F1$ (a1, b1) and $Micro-F1$ (a2, b2) scores for multi-label node classification. $X$-axis: training data ratio}
  \label{Dist_MLC_mac_mic_F1}
\end{figure}

%
%
%
%
%
%
%
%
%
%

\begin{figure}
  \centering
  \includegraphics[width= 3.2 in]{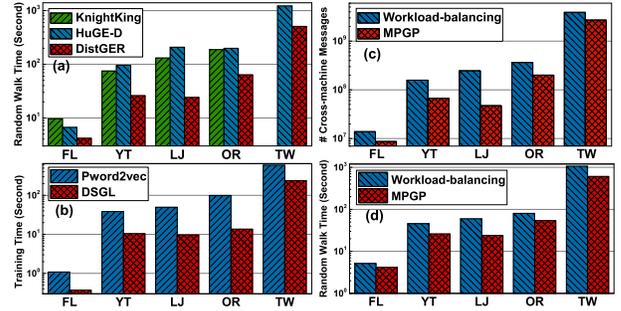}
  \caption{\small {(a) Random walk efficiency, (b) training efficiency, (c) \# cross-machine messages, (d) random walk efficiency for {\sf MPGP} (ours) and workload-balancing scheme ({\sf KnightKing})}}
  \label{Dist_efficiency_sampling_training_MPGP}
\end{figure}

\begin{table*}[t!]
\begin{minipage}{0.275\linewidth}
\centering
\renewcommand\arraystretch{1.2}
\captionof{table}{\small Performance evaluation of partitioning for {\sf DistGER} and Competitors } 
\label{Partition_sechme_overhead}
\begin{scriptsize}
\begin{tabular}{ccccc}
    \multicolumn{5}{c}{{\bfseries (a) Partitioning time for {\sf DistGER} and competitors }} \\
    \hline
    {\bf graph} & {\sf PBG} & {\sf DistDGL} & {\sf DistGER}\\
                &           & ({\sf METIS}) &  ({\sf MPGP}) \\
    \hline
    {\sf FL} & 383.28 s & 127.72 s & \bfseries{15.96 s} \\
    {\sf YT} & 349.15 s & 116.30 s & \bfseries{13.56 s} \\
    {\sf LJ} & 458.52 s & 425.19 s & \bfseries{36.42 s} \\
    {\sf OR} & 2662.62 s & 2761.25 s &\bfseries{294.68 s}\\
    {\sf TW} & 22 hour s & $>$ 1 day &\bfseries{9 hours}\\
    \hline
    \multicolumn{5}{c}{{\bfseries (b) Evaluation of {\sf Parallel MPGP} }} \\
    \hline
    {\bf graph} & {\sf Streaming} & {\sf Partitioning} & {\sf Walking}\\
    \hline
    \multirow{2}{*}{\sf LJ} &DFS+deg & 21.86 s & \bfseries{23.78 s} \\
           & BFS+deg & \bfseries{21.25 s} & 24.79 s \\
    \multirow{2}{*}{\sf OR} &DFS+deg & \bfseries{151.29 s} & 77.12 s \\
           & BFS+deg & 156.37 s & \bfseries{46.55 s} \\
    \multirow{2}{*}{\sf TW} &DFS+deg & \bfseries{1940.65} s & 683.81 s \\
           & BFS+deg & 2034.21 s & \bfseries{590.36 s}
\end{tabular}
\end{scriptsize}
\end{minipage}
\quad
\begin{minipage}{.3\linewidth}
    \centering
    \includegraphics[width= 2.5in, height = 1.45 in]{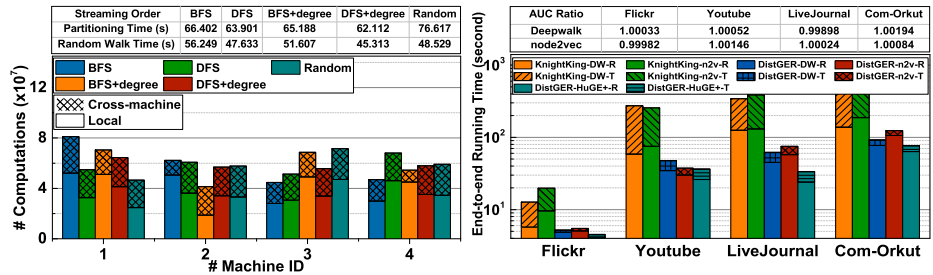}%
    \captionof{figure}
      {\small The distribution of local computations and cross-machine communications for different streaming orders on {\em LiveJournal}. The top table reports their running times for partitioning and random walks
        \label{Dist_MPaD_streaming}
      }
\end{minipage}
\qquad
\begin{minipage}{.37\linewidth}
    \centering
    \includegraphics[width= 2.5 in, height = 1.45 in]{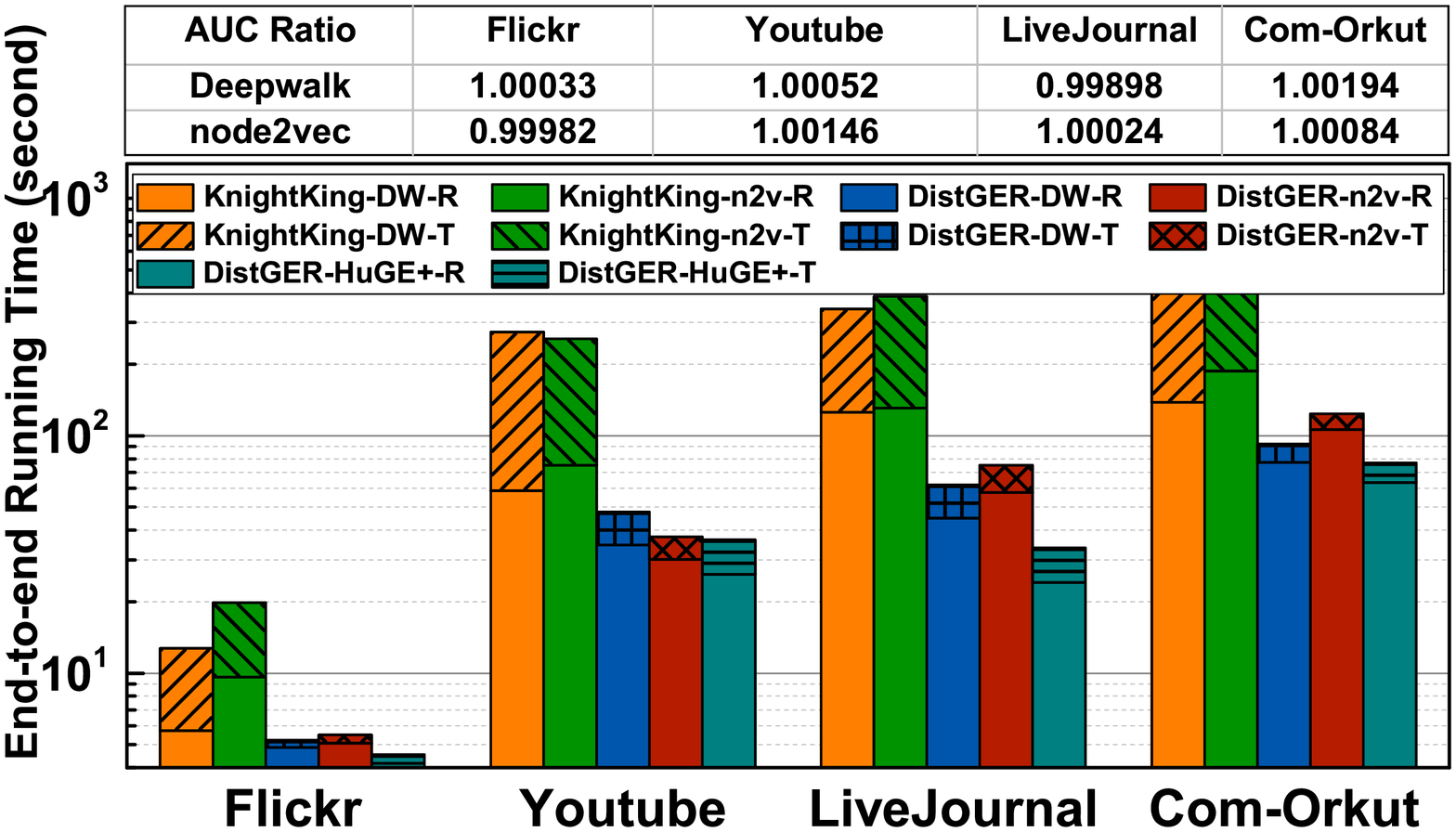}%
    \captionof{figure}
      {\small Generality of {\sf DistGER} vs. {\sf KnightKing}. The bars show random walk efficiency ($-R$) and training efficiency ($-T$) for {\sf Deepwalk} ({\sf DW}), {\sf node2vec} ({\sf n2v}) and {\sf HuGE+}. The top table shows the ratio $\frac{\text{{\em AUC} for {\sf DistGER}}}{\text{{\em AUC} of {\sf KnightKing}}}$, with {\sf DW} and {\sf n2v}, task: link prediction
        \label{Dist_generality}
      }
\end{minipage}
\end{table*}

\subsection{Efficiency due to Individual Parts of DistGER}
\label{sec:individual}
\spara{Random walk and training efficiency.}
To evaluate the system design of {\sf DistGER} (\S \ref{sec:DistGER}, \S \ref{sec:learning}),
we first compare the efficiency of random walks and training with those of {\sf KnighKing} and {\sf HuGE-D}.
For random walks (Figure~\ref{Dist_efficiency_sampling_training_MPGP}(a)),
{\sf DistGER} significantly outperforms {\sf KnightKing} and {\sf HuGE-D} on all our graph
datasets, achieving an average speedup of $3.32\times$ and $3.88\times$, respectively.
Although {\sf HuGE-D} implements information-oriented random walks on {\sf KnightKing},
due to additional computation and communication overheads during on-the-fly information
measurements (\S \ref{sec:HUGED}), its efficiency can be lower than that of {\sf KnightKing}.
We also notice that the random walk lengths ($L$) and the number of random walks ($r$) reduce (on average)
63.2\% and 18\%, respectively, in our information-oriented random walks, compared to {\sf KnightKing}'s
routine random walk configuration.

Another benefit of information-centric random walks is that it generates concise and effective corpus to improve 
training efficiency. Compared to {\sf KnightKing}, {\sf DistGER} achieves $17.37\times$-$27.95\times$ acceleration
in training over all our graphs. Next, considering the same corpus size, we compare the training efficiency of {\sf Pword2vec} and {\sf DSGL}
(trainer in {\sf DistGER}). Figure~\ref{Dist_efficiency_sampling_training_MPGP}(b) shows that {\sf DSGL} achieves $4.31\times$ average speedup
compared to {\sf Pword2vec}. We also notice that the average throughput (number of nodes processed per second) for {\sf DSGL} is up to 49.5 million/s,
while that of {\sf Pword2vec} is only up to 16.1 million/s. These results indicate that our distributed {\sf Skip-Gram} learning model (\S \ref{sec:learning})
is more efficient than {\sf Pword2vec}.
%

\spara{Partitioning efficiency.} Considering the 
randomness inherent in random walks, the partitioning scheme is critical to overall efficiency. 
For {\sf DistGER},
Figure~\ref{Dist_efficiency_sampling_training_MPGP}(c) exhibits that our multi-proximity-aware streaming graph partitioning ({\sf MPGP})
significantly reduces (avg. reduction $45\%$) the number of cross-machine messages than the workload-balancing partition of {\sf KnightKing}
on five graphs. Moreover, it improves the efficiency by 38.9\% for the random walking procedure
(Figure~\ref{Dist_efficiency_sampling_training_MPGP}(d)) over the same set of walks.
We report in Table~\ref{Partition_sechme_overhead}(a) the time required for graph partitioning in competing systems,
where {\sf DistDGL} uses the {\sf METIS} algorithm \cite{METIS_1998} for partitioning.
The results show that {\sf MPGP} performs partitioning with very little overhead in most cases, and
the partitioning efficiency is on average $25.1\times$ faster than competitors.
In Figure~\ref{Dist_MPaD_streaming}, we exhibit the distribution of local computations and cross-machine communications
on four machines for different streaming orders, and the top table reports their running times for partitioning and random walks.
For sequential {\sf MPGP}, we find that the {\sf DFS+degree}-based streaming order (\S \ref{sec:partition}) is more efficient than other streaming orders,
and it also strikes the best balance between cross-machine communications reduction and workload balancing.
Table~\ref{Partition_sechme_overhead}(b) exhibits the performance evaluation of {\sf parallel MPGP} on the small- ({\em LiveJournal}), medium- ({\em Com-Orkut}) and large-scale ({\em Twitter}) graphs. The results show that {\sf DFS+Degree} in {\sf parallel MPGP} is still the best or comparable in terms of partition time, due to the same reason as stated in our third optimization scheme (\S \ref{sec:partition}). On the other hand, {\sf BFS+Degree} in {\sf parallel MPGP} works the best in terms of random walk time due to preserving the locality of the graph structure (our fourth optimization scheme in \S \ref{sec:partition}).
We ultimately recommend {\sf BFS+Degree} for {\sf parallel MPGP}, since it reduces the partition time greatly, while the random walk time is comparable to that obtained from sequential {\sf MPGP}.
%
%
%
%
%
%
%
%

\subsection{Generality of DistGER}
\label{sec:generality}
%
To demonstrate the generality of {\sf DistGER}, we deploy {\sf Deepwalk} \cite{DeepWalk_2014}, {\sf node2vec} \cite{node2vec_2016} and {
\sf HuGE+} \cite{HuGE+_2022}
on {\sf DistGER}. While the original {\sf Deepwalk} and {\sf node2vec} follow
traditional random walks, in {\sf DistGER} the walk length and the number of walks are decided via information-centric measurements.
Next, we also deploy both {\sf Deepwalk} and {\sf node2vec} on {\sf KnightKing} which supports the routine configuration random walk.
Figure~\ref{Dist_generality} illustrates that {\sf DistGER} reduces the random walks time by 41.1\% and 51.6\% on average for
{\sf Deepwalk} and {\sf node2vec}, respectively. For training, {\sf DistGER} is on average $17.7\times$ and $21.3\times$ faster than {\sf KnightKing}+{\sf Pword2vec}
for {\sf Deepwalk} and {\sf node2vec}, respectively.
Moreover, we also show the {\em AUC} ratio of {\sf DistGER} and {\sf KnightKing}, considering {\sf Deepwalk} and {\sf node2vec}, for link prediction.
Our results depict that {\sf DistGER} has comparable (in most cases, higher) {\em AUC} scores, while it improves the efficiency significantly
even for traditional random walk-based graph embedding methods.
{\sf HuGE+} is an extension of {\sf HuGE}, and it uses the same {\sf HuGE} information-centric method to determine the walk length and the number of walks per node. Figure~\ref{Dist_generality} exhibits the compatibility of {\sf HuGE+} on {\sf DistGER} via its general API.
\section{Conclusions}
\label{sec:conclusions}
We proposed {\sf DistGER}, a novel, general-purpose,
distributed graph embedding framework with improved effectiveness, efficiency,
and scalability. {\sf DistGER} incrementally computes information-centric random walks
and leverages multi-proximity-aware, streaming, parallel graph partitioning to achieve high
local partition quality and excellent workload balancing.
{\sf DistGER} also designs distributed {\sf Skip-Gram} learning,
which provides efficient, end-to-end distributed support for node embedding learning.
Our experimental results demonstrated that {\sf DistGER} achieves much better efficiency and
effectiveness than state-of-the-art distributed systems {\sf KnightKing}, {\sf DistDGL}, and {\sf Pytorch-BigGraph},
and scales easily to billion-edge graphs, while it is also generic to support
traditional random walk-based graph embeddings.

\begin{acks}
 This work was supported by the NSFC (No. U22A2027, 61832020, 61821003).
 Arijit Khan acknowledges support from the Novo Nordisk Foundation grant NNF22OC0072415.
 Siqiang Luo acknowledges support from Singapore MOE Tier 2 (T2EP20122-0003), Tier-1 seed (RS05/21) and NTU SUG-NAP.
 Peng Fang acknowledges support from CSC under grant 202106160066.
\end{acks}


\bibliographystyle{ACM-Reference-Format}
\bibliography{ref}


\section{Additional Experimental Results}

\subsection{Directed vs undirected and weighted vs unweighted graphs}

{\sf DistGER} handles undirected and unweighted graphs by default, but can also support directed and weighted (higher edge weights denoting more connectivity strengths) graphs.
Empirically 
since 
FL, YT, LJ, OR, and TW - these widely used graphs are unweighted, following {\sf KnightKing} \cite{KnighKing_2019} we created their weighted versions by assigning edge weights as real numbers uniformly at random sampled from [1, 5), Table \ref{weigthed_unweighted} exhibits end-to-end running times of {\sf DistGER} for all graphs (both unweighted and weighted versions). The running times of weighted versions are slightly higher compared to unweighted versions.
\begin{table}[h]
\vspace{-3mm}
  \caption{End-to-end running times of {\sf DistGER} on unweighted and weighted graphs (seconds)}
  \label{weigthed_unweighted}
  \begin{center}
  \begin{tabular}{ccc}
   \hline
    \bf{End-to-end time} & \bf{Unweighted} & \bf{Weighted}\\
    \hline
    {\bf FL}    &10.038            &11.585\\

    {\bf YT}    &49.982            &52.981\\

    {\bf LJ}    &70.143            &72.598\\

    {\bf OR}    &233.096           &258.966\\

    {\bf TW}    &2779.802          &2890.743\\
  \hline
\end{tabular}
\end{center}
\end{table}

The evaluated graphs in our experiment already included both undirected (FL, YT, and OR) and directed (LJ and TW) graphs. 
Table \ref{directed_undirected} shows running times of various components of {\sf DistGER} on directed and undirected versions of the {\em LiveJournal} (LJ) dataset. Due to the difference in graph structures, 
we observed that the directed version requires more rounds of random walking per node ($r$ defined in Eq. 7) to achieve the convergence in effectiveness measurement (11 and 6 rounds in directed and undirected versions, respectively), and thus the sampling time in the directed version is more than that in the undirected version. In contrast, the directed version, due to less number of edges, has lower training time and also smaller memory footprint as the size of the generated corpus is smaller than that in the undirected version.
%

\begin{table}[h]
\newcommand{\tabincell}[2]{\begin{tabular}{@{}#1@{}}#2\end{tabular}}
\vspace{-2mm}
  \caption{Running times and memory usage of {\sf DistGER} on the directed and undirected versions of the {\em LiveJournal} dataset}
  \label{directed_undirected}
  \begin{center}
  \begin{tabular}{cccc}

    \hline
    { } & \bf{Undirected} &\bf{Directed} \\
    \hline
    {\bf \# Edges} &25,770,074  &1,460,813 \\
    {\bf Partitioning time} &21.257 s  &20.092 s  \\
    {\bf Sampling time} & 24.062 s & 30.790 s  \\
    {\bf Training time} & 9.663 s & 7.005 s \\
    {\bf Memory usage}  &7.43 GB &5.23 GB\\
  \hline
\end{tabular}
\end{center}
\vspace{-6mm}
\end{table}

\subsection{Memory usage}

We evaluated the memory usage of the sampling and training procedure for {\sf DistGER} on each machine (from a cluster of 8 machines) as shown in Table ~\ref{Memory_usage}. Compared with the same type of frameworks, {\sf KnightKing} and {\sf HuGE-D}, {\sf DistGER} requires fewer memory resources during sampling and training. While for the multi-relations- and graph neural network-based frameworks {\sf Pytorch-BigGraph} and {\sf DistDGL} in Table ~\ref{Memory_usage}(b), {\sf DistGER} is comparable and in most cases better than the competitors on the five real-world graphs.

\begin{table}[h!]
\newcommand{\tabincell}[2]{\begin{tabular}{@{}#1@{}}#2\end{tabular}}
  \caption{Avg. memory footprint (GB) of {\sf DistGER} and {\sf competitors} on each machine, where $\sigma$ is the standard deviation, and ``$-$'' means the framework fails under constraints (computation resources or time cost \textgreater $1$ day).}
  \label{Memory_usage}
  \begin{center}
  \renewcommand\arraystretch{1.5}
  \tabcolsep=0.08cm
  \scriptsize
 \begin{tabular}{cccccc}
    \hline
    \multicolumn{6}{c}{\small \bfseries{ (a) Sampling procedure for the random walk-based frameworks}}\\
    \hline
    {\bf{Graph}} &\sf{KnightKing} & \multicolumn{2}{c}{\sf{HuGE-D}} & \multicolumn{2}{c}{\sf{DistGER}} \\
    \hline
     \bf{FK} & 0.66($\pm$0.06)	&  \multicolumn{2}{c}{0.65($\pm$0.03)}     &  \multicolumn{2}{c}{\bf{0.41($\pm$0.02)}}\\

     \bf{YT} &4.11($\pm$0.55)	&  \multicolumn{2}{c}{3.18($\pm$0.39)}     &  \multicolumn{2}{c}{\bf{1.36($\pm$0.23)} }	\\

     \bf{LJ} & 7.65($\pm$0.82)	&  \multicolumn{2}{c}{4.07($\pm$0.72)}     &  \multicolumn{2}{c}{\bf{1.95($\pm$0.16)}	}\\

     \bf{CO} &10.98($\pm$1.03)	&  \multicolumn{2}{c}{6.43($\pm$0.26)}     &  \multicolumn{2}{c}{\bf{3.27($\pm$0.79)}} 	\\

     \bf{TW} & $-$	 	        &  \multicolumn{2}{c}{29.52($\pm$5.96)}     &  \multicolumn{2}{c}{\bf{20.18($\pm$3.62)}} 	\\
  \hline
    \multicolumn{6}{c}{\small \bfseries{ (b) Training procedure for \sf{{DistGER}} and \sf{competitors.}}}\\
    \hline
    {\bf{Graph}} &\sf{KnightKing} & \sf{HuGE-D} &\sf{DistDGL} &\sf{PBG}  & \sf{DistGER}\\
    \hline
     \bf{FK} &1.31($\pm$0.17)	&0.96($\pm$0.10)	&4.91($\pm$0.41)	&8.92($\pm$0.93)    & \bf{0.86($\pm$0.06)} \\

     \bf{YT} &4.73($\pm$0.72) 	&4.41($\pm$0.58)	&5.19($\pm$0.22)	&9.69($\pm$0.95)    & \bf{4.26($\pm$0.63)} 	 \\

     \bf{LJ} & 6.38($\pm$0.97)	&5.63($\pm$0.43)	&7.96($\pm$0.38)	&10.26($\pm$0.87)   & \bf{5.49($\pm$0.85)}	\\

     \bf{CO} &8.52($\pm$1.01)	&7.08($\pm$1.09)	&15.27($\pm$1.12)	&11.37($\pm$1.39)    & \bf{6.86($\pm$0.69) } 	\\

     \bf{TW} & $-$	            &73.72($\pm$7.07)	& $-$	            &\bf{31.65($\pm$4.61)}     & {67.16($\pm$5.18)} \\
  \hline
\end{tabular}
\end{center}
\end{table}

\subsection{Varying $\gamma$}

\begin{figure}[h!]
  \centering
  \includegraphics[width= 3.2 in]{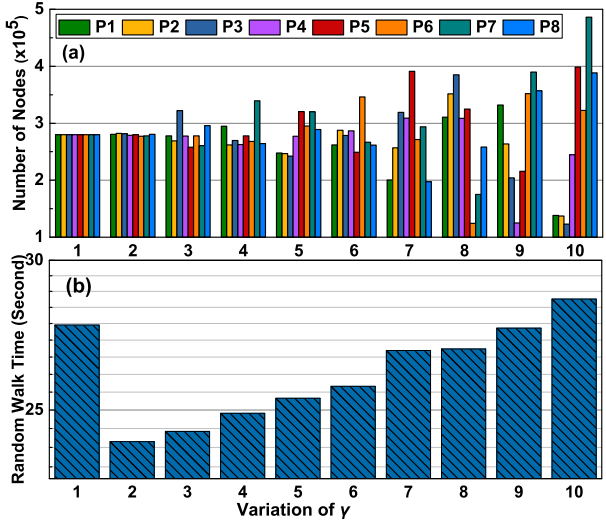}
  \caption{(a) Number of assigned nodes to each partition and (b) average random walk time with varying $\gamma$}
  \label{MPGP_gama}
\end{figure}

$\gamma$ is a slack parameter that allows deviation from the exact load balancing. Figure \ref{MPGP_gama} exhibits how different $\gamma$ affects the quality and load balancing on {\em LiveJournal} (LJ) dataset. We observed that although the partition scheme produces a strict load-balancing at $\gamma=1$, the contribution to the random walk is less due to the low utilization of local partitions, while a larger $\gamma$ (such as $\gamma=10$) introduces skewed workloads for computing machines, which also affects the partition quality and increases the average random walk time. Thus, there is a trade-off between partition quality and load-balancing. We set $\gamma=2$ since it produces the minimum average random walk time.


\subsection{DistGER combined with GPUs}

Each machine in the cluster is equipped with a GPU (NVIDIA GeForce RTX 3090, 24GB memory). We deployed {\sf DistGER}'s learner on GPUs using CUDA V11.2.152.
In Table \ref{DistGER_GPU}, we found that this GPU version (called  {\sf DistGER-GPU}) does not provide a significant improvement compared to {\sf DistGER}. For small graphs, {\sf DistGER-GPU} attains minor speedups; while for large graphs (e.g., {\em Twitter}), the performance is far from impressive due to higher memory consumption of training (refer to Table \ref{Memory_usage}), which exceeds the memory capacity of GPUs, resulting in large data movements between main memory and GPUs.
As we discussed in related work, computing gaps between CPUs and GPUs and the limited memory of GPUs plague the efficiency of graph embedding. 

\begin{table}[h!]
\vspace{-3mm}
  \caption{{\sf DistGER} training time (seconds) in combination with GPU ({\sf DistGER-GPU}) on four computing nodes}
  \label{DistGER_GPU}
  \begin{center}
  \begin{tabular}{ccc}
   \hline
    \bf{Graph} & \bf{DistGER} & \bf{DistGER-GPU}\\
    \hline
    {\bf FL}    &1.791            &0.653\\

    {\bf YT}   &27.529            &20.451\\

    {\bf LJ}    &29.791            &27.835\\

    {\bf OR}    &51.959           &46.341\\

    {\bf TW}   &299.896          &390.081\\
  \hline
\end{tabular}
\end{center}
\vspace{-6mm}
\end{table}

\end{document}